\documentclass[english, onethmnum]{siamltex}

\usepackage[T1]{fontenc}
\usepackage[latin9]{inputenc}
\usepackage{amssymb,amsmath}
\usepackage{graphicx}
\usepackage{units}
\usepackage{babel}

\newcommand{\dint}{\,\mathrm{d}}
\newcommand{\Rn}{\mathbb{R}}
\newcommand{\Nn}{\mathbb{N}}
\newcommand{\Cn}{\mathbb{C}}

\newcommand{\s}{\mathfrak{s}}
\newcommand{\B}{\mathcal{B}}
\newcommand{\Lp}{\mathnormal{L}}
\newcommand{\Nat}{\mathbb{N}}

\newcommand{\op}[1]{\mathrm{#1}}
\newcommand{\eps}{\varepsilon}
\DeclareMathOperator{\sign}{sign}

\begin{document}
\title{Nonlocal actin orientation models select for a unique orientation
  pattern.}
\author{Daniel B. Smith\footnotemark[2]\
  \and Jian Liu \footnotemark[2]} 
\date{\today}
\maketitle

\renewcommand{\thefootnote}{\fnsymbol{footnote}}

\footnotetext[2]{National Heart, Lung and Blood Institute, National Institutes
  of Health, Bethesda, MD, USA}

\renewcommand{\thefootnote}{\arabic{footnote}}

\begin{abstract}
  Many models have been developed to study the role of branching actin networks
  in motility. One important component of those models is the distribution of
  filament orientations relative to the cell membrane. Two mean-field models
  previously proposed are generalized and analyzed. In particular, we find that
  both models uniquely select for a dominant orientation pattern. In the linear
  case, the pattern is the eigenfunction associated with the principal
  eigenvalue. In the nonlinear case, we show there exists a unique equilibrium
  and that the equilibrium is locally stable. Approximate techniques are then
  used to provide evidence for global stability.
\end{abstract}

\section{Introduction}

Actin is a protein involved in many cellular processes ranging from regulating
gene transcription to acting as a motor in cell motility
\cite{Dominguez2011}. It is one of the most conserved proteins and is present
in almost all eukaryotic cells. Actin monomers polymerize into thin filaments
which form highly branched networks near the leading edge of motile cells
\cite{Mullins1998}. While actin monomers will spontaneously polymerize in
physiological conditions, inside these branched networks, new filaments are
generated by branching off of existing filaments \cite{Pollard2000a}. New
filaments are nucleated by the actin related proteins 2 and 3 complex (Arp2/3).
To maintain a consistent supply of actin monomers, actin filaments are
eventually severed and depolymerized. Filament density is regulated by capping
protein binding to the filament tips, ceasing polymerization
\cite{Carlsson2010}. Combined with filaments growing by the addition of new
monomers, these processes create a dynamic network that serves as the engine in
certain types of cell motility \cite{Pollard2000a, Rafelski2004}.

Any individual filament in an actin network can be partially characterized by
the angle between it and the normal direction of the membrane. One obvious
question is whether or not these angles form any regular pattern. While the
question has not been extensively studied experimentally, there is some
evidence that the networks indeed organize into regular patterns relative to
the cell membrane \cite{Maly2001, Weichsel2012, Schaub2007a, Verkhovsky2003}. A
few models have been proposed to explain the existence of such patterns
\cite{Maly2001, Schaus2007, Weichsel2010}. While these models have been
numerically studied, there has been no rigorous work proving the existence,
uniqueness or stability of these solutions. This article presents a few results
that characterize the solutions to two equations modeling the angular density
of branching actin networks.

All of the models proposed to explain the orientation distribution have used a
continuum approximation. There is some question as to whether or not ignoring
the stochastic fluctuations of actin networks is justified
\cite{Rafelski2004}. However, none of the models make specific predictions
about the kinetics of network organization, and there is currently no evidence
that correlations between filaments lead to changes in the equilibrium
orientation pattern. For the rest of this article, we will assume the
approximation is justified and focus on long-time equilibrium behavior.

Some of the first few models to study orientation patterns in actin and similar
networks studied the existence and persistence of peaks in the orientation
pattern \cite{Geigant1998, Mogilner1995, Mogilner1996a}. The analysis was based
on Fourier series and small perturbations which greatly limited their
generality. Their analysis led to the qualitative result that peaked
orientation patterns are likely to be observed. Similar methods have been used
on models of orientation and space \cite{Edelstein-Keshet1989,
  Bressloff2004}. Stability analysis has also been done on similar models,
termed ``ring models'' in the neuroscience literature \cite{Ben-Yishai1995,
  Zhang1996}.

The first model we consider here was proposed by Maly and Borisy
\cite{Maly2001}. Their insight was that if filaments were capped at different
rates based on the filament orientation, filament branching and capping could
generate stable orientation patterns. The model they proposed takes into
account branching and capping explicitly and filament growth implicitly. New
filaments branch off of existing filaments at a characteristic angle
$\sim70^\circ$ with some variance around that. We can write out a branching
kernel as a probability of a mother filament with angle $\theta_M$ having a
branched daughter filament with angle $\theta_D$:
\begin{equation}
  \B(\phi)=\mathrm{P}(\theta_D=\theta_M-\phi)
\end{equation}
  Adding up
the contribution of all filaments with density $u(\phi)$ gives the total
branching rate at angle $\theta$:
\begin{equation}
  \text{BR}(\theta)\propto\int\B(\theta-\phi)u(\phi)\dint\phi
\end{equation}
They also proposed that the capping rate was proportional to the amount of time
the filament would be not in contact with the leading edge, much like
\cite{Mogilner2003b}, but used the capping function $\frac{1}{\cos(\theta)}$.
The Maly and Borisy model only considered filaments growing faster than the
leading edge, i.e. filaments with orientation
$|\theta|\leq\theta_{\text{crit}}=\arccos\left[\frac{v}{v_{max}}\right]$ where
$\frac{v}{v_{max}}$ is the velocity of the leading edge relative to the maximum
velocity of filament growth. Combining the two terms gives the full equation:
\begin{equation}\label{Maly}
  \dot{u}(\theta,t)=\lambda\int\limits_{-\theta_{\text{crit}}}^{\theta_{\text{crit}}}
  \B(\theta-\phi)u(\phi)\dint\phi-\frac{u(\theta,t)}{\cos(\theta)}
\end{equation}
where $\dot{u}$ indicates the time derivative. The equation is defined on
$(-\theta_{\text{crit}},\theta_{\text{crit}})\times\Rn^+$ with absorbing
boundary conditions.

Maly and Borisy performed two analyses on \eqref{Maly}. The first analysis was
to approixmate solutions of \eqref{Maly} by solving the equation for two points
in orientation space. Solutions to the two-point approximation supported the
argument that the equation selected for a unique orientation `type' that grew
exponentially at the fastest rate. The second analysis was to use numerical
quadrature \cite{Baker1977} to approximate the eigenfunctions of the right-hand
side of \eqref{Maly}. However, the existence and uniqueness of the
eigenfunction solutions were never rigorously shown. They explained their
results by using an evolutionary selection metaphor. In this article, we show
that a version of \eqref{Maly} with stricter hypotheses on the capping rate
uniquely selects for a most 'fit' orientation pattern with a fitness function
defined on the unit ball of orientation functions.

A very similar model was proposed by Weichsel and Schwartz \cite{Weichsel2010}
to explain both the orientation patterns and the velocity of a growing actin
network pushing against a given force. There were two primary differences
between their model and the Maly and Borisy model. First, orientations were
defined on the entire circle, $S^1$. The second difference was to normalize the
total branching rate to the constant $\lambda$, which ensures that solutions
have bounded total density. The Weichsel and Schwartz model is:
\begin{equation}\label{Weich}
  \dot{u}(\theta,t)=\frac{\lambda}{\int_{S^1}u(\phi,t)\dint\phi}
  \int\limits_{S^1}\B(\theta-\phi)u(\phi,t)\dint\phi-\kappa(\theta)u(\theta,t)
\end{equation}
where the capping rate is a constant plus a term proportional to the difference
between the velocity of the leading edge and a filament with a given
orientation:
\begin{equation}
  \kappa(\theta)=k+c\left(v_{LE}-v_0\cos(\theta)\right)^+
\end{equation}
where $v_{LE}$ is the velocity of the leading edge, $v_0$ is the rate of
filament growth, and $(x)^+$ is the positive part of $x$.

Weichsel and Schwartz performed the same two analyses as in the Maly and Borisy
paper. They found that, for certain parameters, there were two equilibria in
the two-point approximation to \eqref{Weich}, where one equilibrium is stable
and the other is a saddle. Finally, they used numerical techniques to calculate
the equilibrium distributions. The results in this article explicitly
contradict their assertion of multiple equilibria, but they show the local
stability of a unique, positive equilibrium. However, it is important to note
that the equilibrium is unique for a fixed $\B$ and $\kappa$. Changing the load
force, concentrations of branching and capping proteins, or other experimental
manipulations could change the structure of the unique equilibrium.

The core tool used in this paper to prove the existence and uniqueness of a
principal eigenvalue is the Krein-Rutman theorem. The theorem is one of the key
tools in studying transfer and diffusion operators with applications in biology
\cite{Perthame2006}, physics \cite{Drmac2003}, and materials science
\cite{Capdeboscq2002}. The work presented here is relatively novel in that we
show the equivalence of the spectrum between our operator of interest and a
positive operator before using the Krein-Rutman theorem on the positive
operator.

\subsection{Definitions and Assumptions}

The two equations we specifically analyze are:
\begin{align}\label{first-dyna}
  \dot{u}(\theta,t)&=\lambda\Big(\B\star u\Big)(\theta,t)
  -\kappa(\theta)u(\theta,t)\\\label{zero-dyna}
  \dot{u}(\theta,t)&=\frac{b}{\int u(\omega,t)\dint\omega}
  \Big(\B\star u\Big)(\theta,t)-\kappa(\theta)u(\theta,t)
\end{align}
Equation \ref{first-dyna} is our generalization of the Maly and Borisy
\cite{Maly2001} model, and equation \ref{zero-dyna} is our generalization of
the Weichsel and Schwarz \cite{Weichsel2010} model. Both equations are defined
on the circle $S^1$ with $\B\geq0$ being the branching kernel which generates
new filaments and $\kappa(\theta)>0$ being the variable capping rate which
eliminates filaments. The hypotheses on each function are relatively weak:

\begin{enumerate}
\item $\B$ is real, positive, symmetric $C^2$ function with $\|B\|_1=1$.
\item $\kappa$ is a real, strictly positive, symmetric $C^2$ function
\end{enumerate}

The assumptions are likely stronger than necessary, but generalizing the
problem is a question for further study. They also do not exactly capture the
dynamics for either paper. The first paper, by Maly and Borisy \cite{Maly2001},
would require an infinite capping rate.  However, as that is likely unphysical,
the equations at hand should be sufficient. For Weichsel and Schwarz
\cite{Weichsel2010}, the capping rate they used is continuous but not
differentiable. The primary role of the $C^2$ hypothesis on $\B$ is to ensure
compactness, and weaker hypotheses should be quite feasible. Likewise, the
$C^2$ hypothesis is stronger than necessary, but it simplifies the
presentation. The smoothness hypotheses are merely technical and should have no
effect on the interpretation of the results presented here.

In agreement with the paper \cite{Maly2001}, the first-order branching equation
uniquely selects for an optimal orientation pattern. However, the Weichsel and
Schwarz paper \cite{Weichsel2010} suggests that there might be multiple
equilibria. We show that the zeroth-order branching equation also uniquely
selects for a unique equilibrium orientation pattern.

The first two results characterize solutions to the first-order branching
equation \eqref{first-dyna}. Theorem \ref{principal} shows that the spectrum of
the operator defining \eqref{first-dyna} is dominated by an isolated, simple
principal eigenvalue with strictly positive eigenfunction. While that
eigenvalue may be positive or negative in general, long-time solutions to
\eqref{first-dyna} are therefore dominated by the exponential increase or decay
of the principal eigenfunction. Proposition \ref{small-prop} says that for
given $\B$ and $\kappa$, there exists only one $\lambda$ such that
\eqref{first-dyna} has a non-trivial equilibrium.

The rest of the article is dedicated to analyzing
\eqref{zero-dyna}. Proposition \ref{exist-prop} gives the existence and
uniqueness of a non-trivial equilibrium. Linear stability analysis combined
with Theorem \ref{zero-stable} implies that the equilibrium is locally stable.
Finally, numerical simulations and a perturbation analysis are performed to
provide evidence that \eqref{zero-dyna} is globally stable.

\section{First-order Branching}

The first result uniquely characterizes the dynamics of
\eqref{first-dyna}. Define $\op{A}$ to be the linear operator on the right-hand
side of \eqref{first-dyna}:

\begin{equation}
  \op{A} = \lambda\Big(\B\star u\Big)(\theta,t)-\kappa(\theta)u(\theta,t)
\end{equation}

For the sake of brevity, we will forego much discussion of the existence and
uniqueness of solutions to equation \eqref{first-dyna}. It is known that a
densely-defined resolvent positive operator fulfills the Hille-Yosida
conditions, which ensures unique, positive solutions \cite{Arendt1987}. We will
sketch a quick lemma showing that $\op{A}$ is resolvent positive as it is
illustrative of future techniques:

\begin{lemma}
  $\op{A}$ is a resolvent positive operator, i.e. there exists $\gamma_0\in\Rn$
  such that for all $\gamma\in\Rn$ where $\gamma > \gamma_0$:
  \begin{equation}
    (\gamma - \op{A})^{-1} \geq 0
  \end{equation}
\end{lemma}

\begin{proof}
  For the sake of this sketch, we will avoid the details regarding the
  underlying space $A$ is defined on and will define inequalities pointwise.
  Choose $\gamma_0 = \rho(\op{A})$. Fix $\gamma\in\Rn$ where $\gamma>\gamma_0$
  and some positive function $f\geq 0$. It suffices to show that $g\geq0$ where
  $g$ is:
  \begin{equation}
    (\gamma - \op{A})^{-1}f = g
  \end{equation}
  Working things out, we can observe:
  \begin{align}\nonumber
    (\gamma - \op{A})^{-1}f &= g & \iff\\\nonumber
    f &= (\gamma - \op{A})g & \iff\\\nonumber
    f &= (\gamma + \kappa)g - \B\star g &\iff\\\nonumber
    \frac{f}{\gamma + \kappa} &= g - \frac{\B\star g}{\gamma + \kappa}
  \end{align}
  Define $\op{B}_\gamma = \frac{\B\star}{\gamma+\kappa}$. We can now use the
  Neumann series to finish the proof:
  \begin{align}\nonumber
    g &= (1-\op{B}_\gamma)^{-1}f & \iff \\
    g &= \sum_{k=0}^\infty B_\gamma^k f \geq 0
  \end{align}
  The last inequality comes from our hypothesis that $\B\geq0$.
\end{proof}

The main result we are showing here is as follows:
\begin{theorem}\label{principal}
  $\op{A}$ has an isolated, algebraically simple principal eigenvalue with
  positive eigenfunction.
\end{theorem}

To make things more readable, we will break the proof out into a number of
lemmas and combine them at the end. Much of the analysis in this section relies
on proving facts for $\op{A}$ as an operator on the space $\Lp^2$ and
generalizing to $\Lp^1$. Before doing so, a quick lemma to ensure $\op{A}$ is
bounded on both spaces.

\begin{lemma}
  $\op{A}$ is a bounded linear operator on both $\Lp^1$ and $\Lp^2$.
\end{lemma}

\begin{proof}
  By the hypothesis that $\B$ is bounded, we use Young's inequality to observe:
  \begin{align}\nonumber
    \|\B\star u\|_1&\leq\|\B\|_\infty\|u\|_1\\
    \|\B\star u\|_2&\leq\|\B\|_2\|u\|_2
  \end{align}
  The convolution is therefore a bounded operator on both spaces. Since
  $\op{A}$ is the sum of the convolution and multiplication by a bounded
  function, we can conclude that $\op{A}$ is a bounded operator
\end{proof}

A fact about $\B\star\cdot$ we need:

\begin{lemma}
  The operator $(\B\star u)(\theta)$ is compact.
\end{lemma}

\begin{proof}
When we are considering the operator $\B\star\cdot$ over $\Lp^2$, we can simply
observe that $\B(\theta-\omega)$ is a Hilbert-Schmidt kernel and that implies
that the convolution is compact. However, proving compactness over $\Lp^1$ is
slightly more difficult. We will use the Arzel\`{a}-Ascoli theorem to show that
$\B\star\cdot$ maps bounded sequences to sequences with a convergent
subsequence.

Take a sequence of functions $\{f_n\}_{n\in\Nn}$ where $\|f_n\|_1\leq1$. Using
Young's inequality as above, we obtain a uniform bound on $\|\B\star
f_n\|_\infty$:
\begin{equation}
  \|\B\star f_n\|_\infty\leq\|\B\|_\infty\|f_n\|_1\leq\|B\|_\infty
\end{equation}
Since $\B\in C^2$, we know that $\|\B^\prime\|_\infty<\infty$. We can again apply
Young's inequality to show that the derivative of $\B\star f_n$ is uniformly
bounded:
\begin{equation}
  \|(\B\star f_n)^\prime\|_\infty=\|\B^\prime\star f_n\|_\infty
  \leq\|\B^\prime\|_\infty\|f_n\|_1\leq\|B^\prime\|_\infty
\end{equation}
The $(\B\star f_n)^\prime$ being uniformly bounded implies that $\{\B\star
f_n\}$ is uniformly Lipschitz. That means that Arzel\`{a}-Ascoli holds and
$\B\star\cdot$ is compact.
\end{proof}

Here, we should introduce a bit of notation to clarify which space we are
considering when we talk about the spectrum of $\op{A}$. $\sigma_1(\op{A})$
refers to the spectrum of $\op{A}$ over the space $\Lp^1$, and
$\sigma_2(\op{A})$ is the spectrum over $\Lp^2$. We can now state and prove
a lemma which characterizes the spectrum of $\op{A}$ for much of the complex
plane.

\begin{lemma}\label{eigenvalues}
  All elements of $\sigma_1(\op{A})$ and $\sigma_2(\op{A})$ outside of the
  line $-R(\kappa)=[-\sup\kappa, -\inf\kappa]$ are eigenvalues.
\end{lemma}

\begin{proof}
This result holds equally for all the $\Lp^p$ spaces. I will prove the result
for $\Lp^1$. The argument holds by simply replacing the metric $\|\cdot\|_1$
with $\|\cdot\|_p$. The eigenfunctions over $\Lp^1$ are bounded, so are in all
of the $\Lp^p$ spaces. Fix a number $\mu\in\sigma_1(\op{A})$ with
$\mu\notin[-\sup\kappa,-\inf\kappa]$ and in either the continuous spectrum or
the point spectrum. Since $\mu$ is not in the residual spectrum, we have a
sequence $\{u_n\}_{n\in\Nat}\subset\Lp^1$ with $\|u_n\|_1$ such that:
\begin{equation*}
  \lim_{n\rightarrow\infty}\|(\op{A}-\mu\op{I})u_n\|_1=0
\end{equation*}
By invoking the compactness of $\B\star\cdot$ from the previous step, there
exists a subsequence $u_{n_k}$ such that:
\begin{align*}
  0=\lim_{k\to\infty}\|(\op{A}-\mu\op{I})u_{n_k}\|_1
  &=\lim_{k\to\infty}\|(\B\star u_{n_k})(\theta)
  -(\kappa(\theta)+\mu)u_{n_k}(\theta)\|_1 \\
  &=\lim_{k\to\infty}\|v(\theta)-(\kappa(\theta)+\mu)u_{n_k}(\theta)\|_1
\end{align*}
where $v$ is the limit of $\B\star u_{n_k}$. By hypothesis,
$\kappa(\theta)+\mu\neq0$, so $\displaystyle \frac{1}{\kappa(\theta)+\mu}$ is
bounded. We now have that:
\begin{equation*}
  \lim_{k\to\infty}u_{n_k}(\theta)=\frac{v(\theta)}{\kappa(\theta)+\mu}=w(\theta)
\end{equation*}
almost everywhere. By the fact that $\B\star\cdot$ is closed, $(\B\star
w)(\theta)=v(\theta)$. Applying the above identities shows that
$\op{A}w=\mu w$. Finally, observe that $w\in\Lp^\infty$ since:
\begin{equation}
  w=\frac{\B\star w}{\kappa+\mu}
\end{equation}
$\B\star w$ is bounded by Young's inequality and $\displaystyle
\frac{1}{\kappa(\theta)+\mu}$ is bounded by hypothesis.
\end{proof}

We can show an even stronger correspondence between $\sigma_1(\op{A})$ and
$\sigma_2(\op{A})$.

\begin{lemma}\label{equal-spectra}
  The spectra $\sigma_1(\op{A})$ and $\sigma_2(\op{A})$ are equal outside of
  $[-\sup\kappa, -\inf\kappa]$.
\end{lemma}

\begin{proof}
To show this, we will consider the spectrum in three parts, the point spectrum,
the continuous spectrum, and the residual spectrum. Any eigenvalue of $\op{A}$
on $\Lp^2$ is an eigenvalue on $\Lp^1$ by the inclusion
$\Lp^2\subset\Lp^1$. The reverse inclusion comes from the fact that all of the
eigenvalues over $\Lp^1$ outside of $[-\sup\kappa,-\inf\kappa]$ are in
$\Lp^\infty\supset\Lp^2$. We now have that the two point spectrums are equal.
The result in step 3 implies that there is no elements of the continuous
spectrum outside of $[-\sup\kappa,-\inf\kappa]$, which implies the continuous
spectrums are equal. All that remains is to show that $\op{A}$ has no
residual spectrum on either $\Lp^2$ or $\Lp^1$.

The natural embedding of $\Lp^q$ into $\Lp^p$ where $1\leq p<q\leq\infty$ is a
dense embedding. Continuous functions are dense in $\Lp^1$ as can be seen by
approximating simple functions by continuous functions. Since continuous
functions are in $\Lp^\infty$, that implies $\Lp^\infty$ is dense in
$\Lp^1$. By the inclusion $\Lp^q\subset\Lp^1$, continuous functions are dense
in $\Lp^q$ for all $1<q\leq\infty$.

The last step remains to show that $\op{A}$ has no residual spectrum over
neither $\Lp^2$ nor $\Lp^1$. By self-adjointness, $\op{A}$ has no residual
spectrum over $\Lp^2$. The fact that $\op{A}$ has no residual spectrum over
$\Lp^1$ follows immediately from the density of the embedding $\Lp^2$ in
$\Lp^1$. We know that $\op{A}-\mu\op{I}$ has dense range in $\Lp^2$ for all
$\mu\in\Cn$ whenever $\mu$ is not an eigenvalue. Assume $\mu$ is not an
eigenvalue, the dense embedding and $\op{A}-\mu\op{I}$ having dense range in
$\Lp^2$ implies that $\op{A}-\mu\op{I}\Big|_{\Lp^2}$ is dense in $\Lp^1$. That
implies that $\op{A}-\mu\op{I}\Big|_{\Lp^1}$ is dense in $\Lp^1$ and that $\mu$
is not in the residual spectrum.
\end{proof}

Moving away from the operator theory world for a moment, we need a more set
theoretic lemma for proving that the principal eigenvalue is simple.

\begin{lemma}\label{positive}
  For any given kernel $\B$, there exists $n$ such that
  $\Big(\op{B}_\gamma\Big)^n$ is strongly positive, i.e.:
  \begin{equation*}
    u\geq0\implies\Big(\op{B}_\gamma\Big)^n u>0
  \end{equation*}
  whenever $u\geq0$ is a continuous function not uniformly zero and
  $\gamma>-\inf\kappa$.
\end{lemma}

\begin{proof}
  We are only concerned with whether or not $\Big(\op{B}_\gamma\Big)^n$ is
  positive and not on the specific value of $\Big(\op{B}_\gamma\Big)^n$, so it
  is sufficient to show the result for $\op{B}^n$ where $\op{B}u=\B\star u$.
  Define $\Sigma$ to be the $\sigma-$algebra associated with the Lebesgue
  measure on $S^1$. We can define the set mapping $T:\Sigma\to\Sigma$ as:
  \begin{equation}
    T\Omega = \supp\{\B\star\mathbf{1}_\Omega\}
  \end{equation}
  where $\Omega\in\Sigma$.  Choose some open interval
  $(y-\delta,y+\delta)\subset\supp\{\B\}$ for $y\neq0$,
  $y\in\Rn\backslash\mathbb{Q}$. Assume $\Omega$ contains an open interval
  $(x-\eps,x+\eps)$. Observe that for every $z\in(x+y-\delta,x+y+\delta)$:
  \begin{equation}
    \B\star\mathbf{1}_\Omega(z)=\int\B(z-s)\mathbf{1}_\Omega(s)\dint s
    \geq\int\limits_{x-\eps}^{x+\eps}\B(z-s)\dint s
    =\int\limits_{y-\eps}^{y+\eps}\B\Big((z-x+y)-s^\prime\Big)\dint s^\prime>0
  \end{equation}
  The above argument also holds for $z\in(x-y-\delta,x-y+\delta)$. Notice that
  while the existence of the interval was used in the above calculation, there
  is no explicit dependence on $\eps$ beyond that $\eps>0$. Iterating $T$,
  we can observe that:
  \begin{equation}
    T^{2n}\Omega\subset\cup_{0\leq j\leq n}\Big((x+2jy-\delta,x+2jy+\delta)
    \cup(x-2jy-\delta,x-2jy+\delta)\Big)
  \end{equation}
  By the fact that $x+2y$ is an irrational
  rotation, $\{x+2jy\}_{j\in\Nat}$ is dense in $S^1$ and the sets
  $\{(x+2jy-\delta,x+2jy+\delta)\}_{j\in\Nat}$ form an open covering of
  $S^1$. The compactness of $S^1$ implies that there exists $n\in\Nat$ such
  that $S^1\subset\cup_{1\leq j\leq n}(x+2jy-\delta,x+2jy+\delta)$. Rotational
  symmetry in $S^1$ implies that $n$ has no dependence on $x$. Fix some
  function $u\in C(S^1)$ with $u\geq0$ and $u$ not uniformly zero.  We can set
  $\Omega=\supp\{u\}$ and observe $\Omega$ contains an open interval
  containing some $x^\prime$. The above discussion implies:
  \begin{equation}
    T^{2n}\Omega\supset\cup_{1\leq j\leq n}
    (x^\prime+2jy-\delta,x^\prime+2jy+\delta)
    \supset S^1
  \end{equation}
  The above set relation implies $\op{B}^{2n}u>0$.
\end{proof}

We can now show the main result of this section, Theorem \ref{principal}.

\begin{proof}
  Since $\op{A}$ is self-adjoint, we know the spectrum over $\Lp^2$ is bounded
  by the eigenvalue $\mu_0$:
  \begin{equation}
    \mu_0=\sup_{\|u\|_2=1}\langle\op{A}u,u\rangle
  \end{equation}
  By lemma \ref{eigenvalues}, we know that $\mu_0$ is an eigenvalue as long as
  $\mu_0>-\inf\kappa$. It suffices to show there exists $u\in\Lp^2$ such that:
  \begin{equation}
    \langle\op{A}u,u,\rangle+\inf\kappa\langle u,u\rangle>0
  \end{equation}
  By the continuity of $\kappa$ and compactness of the circle,
  $\kappa(\theta)-\inf\kappa=0$ for at least one $\theta$. For the sake of
  notation, define $g(\theta)=\kappa(\theta)-\inf\kappa$.

  Observe that for any function of the form $u(\theta)=c+f(\theta)\geq0$ with
  $c,\,f(\theta)\geq0$ where $\int_{S^1}u(\theta)=1$:
  \begin{align*}
    \langle\B\star u,u\rangle
    &=\iint\limits_{S^1S^1}\B(\theta-\omega)(c+f(\omega))\dint\omega
    (c+f(\theta))\dint\theta \\
    &\geq c\iint\limits_{S^1S^1}\B(\theta-\omega)(c+f(\theta))
    \dint\omega\dint\theta
    =c
  \end{align*}
  since $f(\theta)\geq0$ and $\int\B=1$.  Without loss of generality, assume
  $g(0)=0$. By our hypothesis that $\kappa$ is $C^2$ and $0$ is a local minima,
  we have the inequality:
  \begin{equation}
    g(\theta)=g(\theta)-g(0)\leq Q|\theta-0|^2=Q\theta^2
  \end{equation}
  where $Q = \sup|g^{\prime\prime}|$. Define $f_\eps$ as:
  \begin{equation*}
    f_\eps=\frac{\mathbf{1}_{[-\eps,\eps]}}{2\eps}
  \end{equation*}
  where $\mathbf{1}$ is the usual indicator function. Note that $\int f=1$. We
  now have the two relations:
  \begin{align*}
    \int\limits_{S^1}f_\eps(\theta)g(\theta)\dint\theta
    &=\int\limits_{-\eps}^\eps f_\eps(\theta)g(\theta)\dint\theta
    \leq\frac{2}{\eps}\int\limits_0^\eps Q\theta^2\dint\theta
    =Q\frac{2\eps^2}{3} \\
    \int\limits_{S^1}f_\eps(\theta)^2g(\theta)\dint\theta
    &=\int\limits_{-\eps}^\eps f_\eps(\theta)^2g(\theta)\dint\theta
    \leq\frac{2}{\eps^2}\int\limits_0^\eps Q\theta^2\dint\theta
    =Q\frac{2\eps}{3}
  \end{align*}
  Combining those relations gives:
  \begin{align*}
    \langle g(c+f_\eps),(c+f_\eps)\rangle
    &=\int\limits_{S^1}g(\theta)(c+f_\eps(\theta))^2\dint\theta\\
    &=\int\limits_{S^1}g(\theta)\Big(c^2+2cf_\eps(\theta)+f_\eps(\theta)^2\Big)
    \dint\theta
    \leq Rc^2+cQ\frac{4\eps^2}{3}+Q\frac{2\eps}{3}
  \end{align*}
  where $R=\int g(\theta)\dint\theta$.

  First, assume $R\geq2\pi$ and $Q\geq1$. Fix $c=\frac{1}{2R}$,
  $\eps=\frac{1}{4RQ}$ and $c^\prime=1-2\pi c<1$. Fix $u=c+c^\prime
  f_\eps$. Putting all of the above together gives:
  \begin{align*}
    \langle\op{A}u+(\inf\kappa)u,u\rangle
    &=\langle\B\star(c+c^\prime f_\eps),c+c^\prime f_\eps\rangle
    -\langle g(c+c^\prime f_\eps),c+c^\prime f_\eps\rangle \\
    &\geq c-Rc^2-cc^\prime Q\frac{4\eps^2}{3}-c^{\prime2}Q\frac{2\eps}{3} \\
    &\geq\frac{1}{2R}-\frac{R}{4R^2}-\frac{Q}{R}\frac{2\eps^2}{3}
    -Q\frac{2\eps}{3} \\
    &\geq\frac{1}{4R}-\frac{Q}{3R}\frac{1}{8R^2Q^2}-\frac{Q}{6RQ} \\
    &\geq\frac{1}{4R}-\frac{1}{24R}-\frac{1}{6R}=\frac{1}{24R}>0
  \end{align*}
  If $R<2\pi$, then set $c=\frac{1}{2\pi}$ and $c^\prime=0$. If $R\geq2\pi$ and
  $Q<1$, set $c$ and $c^\prime$ as above and $\eps=\frac{1}{4R}$. That shows we
  have constructed such a $u$ and $\max\sigma_2(\op{A})>-\inf\kappa$.

  Now that we have a principal eigenvalue, to show that it is isolated, take
  a sequence $\mu_j\to\mu_0$ with associated eigenfunctions $u_j$. Choose a
  subsequence such that $\B\star u_{j_k}$ is convergent:
  \begin{align}\nonumber
    0&=\lim_{k\to\infty}\|\B\star u_{j_k}-\B\star u_{j_{k+1}}\|_2 \\\nonumber
    &=\lim_{k\to\infty}\|(\kappa+\mu_{j_k})u_{j_k}-
    (\kappa+\mu_{j_{k+1}})u_{j_{k+1}}\|_2
    \\\nonumber
    &=\lim_{k\to\infty}\|(\kappa+\mu_0)(u_{j_k}-u_{j_{k+1}})\|_2 \\
    &\leq\lim_{k\to\infty}(\inf\kappa+\mu_0)\|u_{j_k}-u_{j_{k+1}}\|_2
    =(\inf\kappa+\mu_0)\sqrt{2}>0
  \end{align}
  since $u_j$ and $u_{j_k}$ are orthogonal and $\mu_0>-\inf\kappa$.

  The last thing to show is that the eigenvalue is simple with positive
  eigenfunction. From the existence of an eigenvalue for $\op{A}$, we know
  that the following eigenvalue equation has at least one solution:
  \begin{equation}\label{pos-eigen}
    \frac{\B\star u}{\kappa+\mu_0}=u
  \end{equation}
  $\op{B}_{\mu_0} u=\frac{\B\star u}{\kappa+\mu_0}$ is obviously a positive and
  compact operator on the Banach space of continuous functions. The
  Krein-Rutman theorem implies that $\op{B}_{\mu_0}$ has an eigenvalue equal to
  its spectral radius. Assume that spectral radius $\rho(\op{B}_{\mu_0})>1$:
  \begin{align}\nonumber
    \frac{\B\star u}{\kappa+\mu_0}&=\rho(\op{B}_{\mu_0}) u & \iff \\\nonumber
    \B\star u-(\kappa u+\mu_0)u&=(\rho(\op{B}_{\mu_0})-1)(\kappa+\mu_0)u & \iff\\
    \op{A}u-\mu_0u &= (\rho(\op{B}_{\mu_0})-1)(\kappa+\mu_0)u &
  \end{align}
  That last equality implies:
  \begin{equation}
    \langle\op{A}u-\mu_0u,u\rangle = (\rho(\op{B}_{\mu_0})-1)\langle
    (\kappa+\mu_0)u,u\rangle>0
  \end{equation}
  in contradiction to our definition of $\mu_0$. Therefore,
  $\rho(\op{B}_{\mu_0})=1$, and there exists a positive $u$ that solves
  \eqref{pos-eigen}.

  Now it remains to show that the eigenvalue is simple. The Krein-Rutman
  theorem implies that it is sufficient to show that
  $\Big(\op{B}_{\mu_0}\Big)^n$ is strongly positive for some $n$. That is shown
  in Lemma \ref{positive}. With Lemma \ref{positive}, we have that
  $\Big(\op{B}_{\mu_0}\Big)^n$ is a strongly positive operator with leading
  eigenvalue 1. Since $\op{B}_{\mu_0}$ and $\Big(\op{B}_{\mu_0}\Big)^n$ have the
  same eigenvalues, $\Big(\op{B}_{\mu_0}\Big)^n$ having a simple leading
  eigenvalue implies the leading eigenvalue of $\op{B}_{\mu_0}$ is also simple.
\end{proof}

Another small proposition to characterize solutions to \eqref{first-dyna}:
\begin{proposition}\label{small-prop}
  Given $\B$ and $\kappa$, there exists precisely one $\lambda$ such that
  \eqref{first-dyna} has a stable, non-trivial equilibrium.
\end{proposition}
Define the operator $\op{A}_\lambda$ as:
\begin{equation*}
  \op{A}_\lambda u=\lambda\B\star u -\kappa u
\end{equation*}
Define the related operator and inner product spaces:
\begin{equation*}
  \op{A}^\prime u=\frac{\B\star u}{\kappa}
  \qquad
  \langle f,g\rangle_{\kappa}
  =\int_{S^1}f(\theta)g(\theta)\kappa(\theta)\dint\theta
\end{equation*}
We can now prove the result.
\begin{proof}
  Observe $\op{A}^\prime$ is self-adjoint with respect to
  $\langle\cdot,\cdot\rangle_\kappa$. Also, we have the relation between
  $\op{A}_\lambda$ and $\op{A^\prime}$:
  \begin{align*}
    \langle\op{A}_\lambda u,u\rangle
    &=\iint\limits_{S^1S^1}\lambda\B(\theta-\omega)u(\omega)\dint\omega\,
    u(\theta)-\kappa u(\theta)^2\dint\theta \\
    &=\iint\limits_{S^1S^1}\lambda\frac{\B(\theta-\omega)}{\kappa(\theta)}
    u(\omega)\dint\omega\,u(\theta)\kappa(\theta)\dint\theta
    -\int_{S^1}u(\theta)^2\kappa(\theta)\dint\theta \\
    &=\lambda\langle\op{A}^\prime u,u\rangle_\kappa-\langle u,u\rangle_\kappa
  \end{align*}
  From the proof of Theorem \ref{principal}, we have that $\op{A}^\prime$ has a
  simple principal eigenvalue, $\mu_0^\prime>0$. We know that $\mu_0$ and
  $\mu_0^\prime$ can be defined by the following:
  \begin{equation}
    \mu_0=\sup_{\|u\|_2\neq0}\frac{\langle\op{A}u,u\rangle}{\langle u,u\rangle}
    \qquad
    \text{and}
    \qquad
    \mu_0^\prime=\sup_{\langle u,u,\rangle_\kappa\neq0}
    \frac{\langle\op{A}^\prime u,u\rangle_\kappa}{\langle u,u\rangle_\kappa}
  \end{equation}
  There is a relationship between the sign of $\mu_0$ and $\mu_0^\prime$:
  \begin{align}\nonumber
    \sign[\lambda\mu_0^\prime-1]&=
    \sign\left[\sup_{\langle u,u\rangle_\kappa\neq0}
      \frac{\lambda\langle\op{A}^\prime u,u,\rangle_\kappa
        -\langle u,u,\rangle_\kappa}{\langle u,u\rangle_\kappa}\right]\\\nonumber
    &=\sup_{\langle u,u\rangle_\kappa\neq0}\left[\sign
    \frac{\lambda\langle\op{A}^\prime u,u,\rangle_\kappa
      -\langle u,u,\rangle_\kappa}{\langle u,u\rangle_\kappa}\right]\\\nonumber
    &=\sup_{\|u\|_2\neq0}\left[\sign
      \frac{\langle\op{A}_\lambda u,u\rangle}{\langle u,u\rangle}\right] \\
    &=\sign\left[\sup_{\|u\|_2\neq0}
      \frac{\langle\op{A}_\lambda u,u\rangle}{\langle u,u\rangle}\right]
    =\sign[\mu_0]
\end{align}
The argument holds since $\frac{\langle u,u\rangle_\kappa}{\langle
  u,u\rangle}>0$ by hypothesis of the supremum and does not change the sign of
the argument. We know that \eqref{first-dyna} has a non-trivial equilibrium if
and only if $\op{A}_\lambda$ has a zero eigenvalue. That equilibrium is stable
if and only if all of the other elements of the spectrum have negative real
part.  That is the case if and only if the zero eigenvalue is the largest
eigenvalue, i.e. $\mu_0=0$. The calculation above therefore implies there
exists only one $\lambda$ where \eqref{first-dyna} has a stable equilibrium
since there is only on $\lambda$ such that $\mu_0=\lambda\mu_0^\prime-1=0$.
\end{proof}

\section{Zeroth-order Branching}

\subsection{Existence of Solutions}

Proving the existence of solutions to \eqref{zero-dyna} does not require any
sophisticated machinery. Define $\op{G}(u)$ to be the nonlinear operator that
defines the dynamics of \eqref{zero-dyna}. First, to show local existence, we
need that $\op{G}$ is locally Lipschitz:
\begin{lemma}
  $\op{G}$ is locally Lipschitz for all $u\in\Lp^1$ where $u\geq0$.
\end{lemma}

\begin{proof}
The derivative of $\op{G}$ is equal to:
\begin{equation}\label{deriv}
  D\op{G}_u(v)=\frac{(\B\star v)(\theta)}{\int_{S^1}u(\omega)\dint\omega}
  -\kappa(\theta)v(\theta) -\int\limits_{S^1}v(\omega)\dint\omega
  \frac{(\B\star u)(\theta)}{\left(\int_{S^1}u(\omega)\dint\omega\right)^2}
\end{equation}
Fix $\eps<\frac{\|u\|_1}{2}$. For all $v\in B(u, \eps)$, the $\Lp^1$-norm ball
around $u$, we have the following inequality:
\begin{equation}
  \left\|D\op{G}_v\right\|_{op}\leq\frac{\|v\|_1}{\|u\|_1-\eps}
  +\sup\kappa\|v\|_1+\frac{\|v\|_1^2}{(\|u\|_1-\eps)^2}
\end{equation}
\end{proof}

The above lemma with the standard Picard-Lindel\"{o}f argument is sufficient to
show local existence on $\Lp^1\times\Rn^+$. To show global existence, we need
to show that $\op{G}(u)$ is uniformly Lipschitz on its domain. Define the
closed set $\mathcal{U}(c, c^\prime)$ to be:
\begin{equation*}
  \mathcal{U}(v):=\{v\in\Lp^1
    :v(\theta)\geq0\;\text{and}\;0<c\leq\|v\|_1\leq c^\prime<\infty\}
\end{equation*}

\begin{lemma}
  $\op{G}$ is uniformly Lipschitz on $\mathcal{U}(c, c^\prime)$ for every
  $0<c<c^\prime<\infty$.
\end{lemma}
\begin{proof}
It is easy to see that a coarse estimate for the supremum of the operator norm
is:
\begin{equation*}
  \|D\op{G}_u\|_{op}\leq\frac{2}{c}+\|\kappa\|_\infty
\end{equation*}
which implies that $\op{G}$ is Lipschitz on $\mathcal{U}(c,c^\prime)$.
\end{proof}

Global existence for all $t\geq0$ can be shown by observing that solutions with
positive, integrable initial data stay in $\mathcal{U}(c,c^\prime)$ for some
$c$, $c^\prime$.

\begin{lemma}
  Given initial data $v(\theta)\geq0$, $v\in\Lp^1$, there exists
  $0<c<c^\prime<\infty$ such that solutions to \eqref{zero-dyna} stay in
  $\mathcal{U}(c, c^\prime)$
\end{lemma}

\begin{proof}
It is obvious that solutions with positive initial data
remain positive. Simply observe for any angle $\theta^\star$ with
$u(\theta^\star,t)=0$ where $u(\cdot,t)\geq0$:
\begin{equation}
  \dot{u}(\theta,t)=(\B\star u)(\theta,t)-\kappa(\theta)u(\theta,t)
  =(\B\star u)(\theta,t)\geq0
\end{equation}
It is also straightforward to show that there exists $c$ and $c^\prime$ for the
definition of $\mathcal{U}$. First, observe that
\begin{equation}
  \frac{\partial}{\partial t}\int u(\theta,t)\dint\theta
  =\lambda_0-\int \kappa(\theta)u(\theta,t)\dint\theta
\end{equation}
The mean value theorem gives:
\begin{equation}
  (\inf\kappa)\int u(\theta,t)\dint\theta
  \leq\int\kappa(\theta)u(\theta,t)\dint\theta
  \leq(\sup\kappa)\int u(\theta,t)\dint\theta
\end{equation}
We can then write out explicit expressions for $c$ and $c^\prime$:
\begin{equation}
  c=\min\left\{\int u(\theta,0)\dint\theta,\,
  \frac{\lambda_0}{\sup\kappa}\right\}
  \qquad
  c^\prime=\max\left\{\int u(\theta,0)\dint\theta,\,
  \frac{\lambda_0}{\inf\kappa}\right\}
\end{equation}
\end{proof}

Picard-Lindel\"{o}f argument is now sufficient to
show global existence, and we have the following result, stated without proof:

\begin{proposition}
  Given initial data $v(\theta)\geq0$, $v(\theta)\in\Lp^1$. There exists
  $u(\theta, t)$ defined on $\Lp^1\times\Rn^+$ where $u(\theta, 0)=v(\theta)$
  and $\frac{\partial}{\partial t}u(\theta, t) = \op{G}(u(\theta, t))$
\end{proposition}

\subsection{Existence of a Unique Equilibria}

The first result is an existence result:
\begin{proposition}\label{exist-prop}
  A function $u\in\Lp^1$ is an equilibrium of equation \eqref{zero-dyna} if and
  only if it is a solution to the eigenvalue problem:
  \begin{equation}\label{zero-eigen}
    \frac{\B\star u}{\kappa(\theta)}=\mu u(\theta)
  \end{equation}
  where $\mu\neq0$ and $\int_{S^1}u(\omega)\dint\omega\neq0$.
\end{proposition}
Proposition \ref{exist-prop} is in contradiction to the hypothesis in
\cite{Weichsel2010} that there are multiple equilibrium solutions to
\eqref{zero-dyna}.

\begin{proof}
  Assume you have an equilibrium $u\in\Lp^1$ with $\int
  u(\omega)\dint\omega\neq0$, i.e. $\op{G}(u)=0$. We know that:
  \begin{equation*}
    \frac{(\B\star u)(\theta)}{\int_{S^1}u(\omega)\dint\omega}
    -\kappa(\theta)u(\theta)=0
  \end{equation*}
  Simple algebra gives:
  \begin{equation*}
    \frac{(\B\star u)(\theta)}{\kappa(\theta)}
    =\int\limits_{S^1}u(\omega)\dint\omega\,u(\theta)
  \end{equation*}
  That implies $u$ is an eigenfunction with eigenvalue $\int
  u(\omega)\dint\omega\neq0$. For the other direction, assume that we have:
  \begin{equation*}
    \frac{\B\star u}{\kappa} = \mu u
  \end{equation*}
  and the listed hypotheses above. Simple algebra again:
  \begin{equation*}
    \B\star u - \mu\kappa u = 0
  \end{equation*}
  Assume $\int_{S^1}u(\omega)\dint\omega=1$. This is justified as long as
  $\int_{S^1}u(\omega)\dint\omega\neq0$.

  The existence of at least one positive eigenfunction with non-zero integral
  is ensured by the Krein-Rutman theorem as in Theorem
  \ref{principal}. $v=\mu_0 u_0$ is now an equilibrium to \eqref{zero-dyna}.
  Lemma \ref{positive} ensures that $u>0$.  By the self-adjointness of
  $\op{A}^\prime$, we have the all eigenfunctions $u_k\neq u_0$ are orthogonal to
  $u_0$:
  \begin{equation*}
    0=\langle u_0,u_k\rangle_\kappa
    =\int\limits_{S^1}u_0(\theta)u_k(\theta)\kappa(\theta)\dint\theta
  \end{equation*}
  However, we know that $\kappa u_0>0$. That implies the above can only be true
  if $u_k\equiv0$ almost everywhere or $u_k$ is negative on some set with
  non-zero measure.
\end{proof}

\subsection{Local Stability}

The next result implies stability of \eqref{zero-dyna} in a local sense. Before
the proof, a quick, basic lemma from complex analysis:
\begin{lemma}\label{complex}
  Assume $x\in\Rn$ and $0<x\leq y$. There exists a real function
  $\beta(\alpha,y)>0$ for $\alpha\neq0\in\Rn$ such that:
  \begin{equation}
    \frac{1}{|x+i\alpha|}\leq\frac{1}{x+\beta}
  \end{equation}
\end{lemma}

\begin{proof}
  Direct calculation shows:
  \begin{align}\nonumber
    \frac{1}{|x+i\alpha|}&\leq\frac{1}{x+\beta} & \iff \\\nonumber
    (x+\beta)^2&\leq x^2+\alpha^2 & \iff \\\nonumber
    0&\leq x^2+\alpha^2-(x+\beta)^2 & \iff \\
    0&\leq \alpha^2-\beta^2-2x\beta
  \end{align}
  Choosing $\beta(\alpha,y)=\min\left\{1,\frac{\alpha^2}{1+2y}\right\}$
  completes the proof.
\end{proof}

\begin{theorem}\label{zero-stable}
  The spectral bound of the linearization around $u_0$, the equilibrium of
  \eqref{zero-dyna}, is strictly less than zero.
\end{theorem}
Define the spectral bound to be:
\begin{equation}
  \s(\op{D})=\sup\{\Re\mu:\mu\in\sigma(\op{D})\}
\end{equation}
We will show that the right half of the complex plane is contained in the
resolvent. Define $\op{D}=D\op{G}_{u_0}$. For these purposes, we will only
consider $\op{D}$ as an operator over $\Lp^2$, but the results are immediately
generalizable to $\Lp^1$ using the techniques in the proof of Theorem
\ref{principal}. For all $\gamma$ such that $\Re\gamma\geq0$ and $\gamma\neq0$,
the proof will consist of constructing a Neumann-type series and showing the
series converges to a bounded operator. The proof is completed by showing
$0\notin\sigma(D\op{G}_{u_0})$.

\begin{lemma}
  The line ${\gamma\in\Rn:\,\gamma>0}$ is in $\rho(\op{D})$.
\end{lemma}

\begin{proof}
  First, fix $\gamma>0$ with $\gamma\in\Rn$.  We can solve explicitly for the
  resolvent. Fix $f$ in $\Lp^2$ and assume there exists $v$ such that:
  \begin{equation}
    \op{A}v-\gamma v = f
  \end{equation}
  We will derive a Neumann-type series to show the existence of such a $v$.
  Expanding out $\op{D}$ and rearranging gives:
  \begin{equation}\label{res-rearrange}
    v-\frac{\B\star v}{\kappa+\gamma}=-\frac{f}{\kappa+\gamma}
  \end{equation}
  From previous results, it is easy to see that:
  \begin{equation}
    r\left(\frac{\B\star}{\kappa+\gamma}\right)<
    r\left(\frac{\B\star}{\kappa}\right)=1
  \end{equation}
  where $r(\cdot)$ is the spectral radius. Define
  $\op{B}_\gamma=\frac{\B\star}{\kappa+\gamma}$. The usual Neumann series gives
  us the explicit form for $(\op{A}-\gamma)^{-1}$:
  \begin{equation}\label{Neumann}
    v=-\sum_{j=0}^\infty\op{B}_\gamma^j \frac{f}{\kappa+\gamma}
  \end{equation}
  The convergence of the above series implies that $(\op{D}-\gamma)^{-1}$ is a
  bounded operator and $\gamma\in\rho(\op{D})$, the resolvent.
\end{proof}

The above argument holds equally well for $\gamma\in\Cn$ where $\Re\gamma\geq0$
and $\gamma\neq0$.

\begin{lemma}
  The set ${\gamma\in\Cn:\,\gamma\neq0 and \Re\gamma\geq0}$ is in
  $\rho(\op{D})$.
\end{lemma}
\begin{proof}
  Equation \eqref{res-rearrange} is equally valid for complex $\gamma$. Assume
  $\Re\gamma\geq0$ and $\gamma\neq0$. In order to show that \eqref{Neumann}
  still holds, we need to show that the spectral radius is strictly less than
  one. Using Lemma \ref{complex}, we can provide a bound on the numerical
  radius $n(\cdot)$:
  \begin{align}\nonumber
    n(\op{B}_\gamma)=\sup_{\|v\|_2=1}|\langle\op{B}v,v\rangle|
    &=\sup_{\|v\|_2=1}\left|\int\frac{(\B\star v)(\theta)}{\kappa(\theta)+\gamma}
    v(\theta)\dint\theta\right|\\\nonumber
    &\leq\sup_{\|v\|_2=1}\int
    \frac{|(\B\star v)(\theta)v(\theta)|}{|\kappa(\theta)+\gamma|}\dint\theta
    \\
    &\leq\sup_{\|v\|_2=1}\int
    \frac{|(\B\star v)(\theta)v(\theta)|}{\kappa(\theta)+\beta}\dint\theta<1
  \end{align}
  The Neumann series in equation \eqref{Neumann} therefore converges and
  $\gamma\in\rho(\op{D})$ for all $\gamma\neq0$ where $\Re\gamma\geq0$.
\end{proof}

The last remaining case to prove is that $0\in\rho(\op{D})$.
\begin{lemma}
  $0\in\rho(\op{D})$
\end{lemma}

\begin{proof}
  Assume $0\in\sigma(D\op{G}_{u_0})$. That would imply that:
  \begin{equation}
    \B\star v-\kappa v=\left(\int v(\omega)\dint\omega\right)\kappa u_0
  \end{equation}
  The nullspace of the left hand side is spanned by $u_0$. Since $v=u_0$ does
  not solve the equation, we must have $\int v\neq0$. However, the Fredholm
  alternative states that the above is only solvable if the right hand side is
  perpendicular to $u_0$. As $\kappa u_0$ is not perpendicular to $u_0$, the
  equation is not solvable.
\end{proof}

Combining the three lemmas proves Theorem \ref{zero-stable}.

\section{Approximate Methods}

The stability result in Theorem \ref{zero-stable} is a local result, so
approximate methods were used to characterize the global behavior of
\eqref{zero-dyna}. First, a set of numerical simulations were performed that
showed global exponential stability. Then, we present a perturbation expansion
that provides evidence towards global stability along with a discussion of the
limitations of the perturbation.

\subsection{Numerical Simulations}

The two branching kernels were:
\begin{gather} \nonumber
  \B_1(\theta)=\frac{\phi(\theta+\frac{\pi}{2})+\phi(\theta-\frac{\pi}{2})}{2}
  \\\nonumber
  \text{and} \\
  \B_2(\theta)\propto
  \begin{cases}
    1-\theta^6+3\theta^4-3\theta^2 & \text{if}\;|\theta|<1 \\
    0                             & \text{if}\;|\theta|\geq1
  \end{cases}
\end{gather}
where $\B_2$ was normalized to have integral one and $\phi(\theta)$ was a von
Mises distribution:
\begin{equation}
  \phi(\theta)=\frac{\exp\left[\sigma^{-2}\cos(\theta)\right]}
  {2\pi I_0(\sigma^{-2})}
\end{equation}
where $\sigma=\frac{7\pi}{180}$. $\B_1$ was based on the branching kernels used
in \cite{Maly2001} and \cite{Weichsel2010}. Those two papers used truncated
Gaussian distributions centered around $\pm70^\circ$. Here, the von Mises
distribution was used to avoid truncating the Gaussian or using the more
complicated, formally correct wrapped Gaussian distribution. Also, the offset
of $\pm\frac{\pi}{2}$ was used to simplify the radians conversion. Finally, the
constant $\sigma$ was chosen to be in line with previous numerical studies
\cite{Schaus2007, Schaus2008, Weichsel2010}. Both branching kernels were $C^2$
and symmetric.

The two capping functions used were:
\begin{equation}
  \kappa_1(\theta) = 1 - \frac{1}{2}\cos(\theta)
  \qquad
  \text{and}
  \qquad
  \kappa_2(\theta) = 1 + \frac{3}{4}\cos(4\theta^2)
\end{equation}
The first capping function, $\kappa_1$ was based upon \cite{Weichsel2010}, and
the second was chosen to have multiple minima and maxima and non-uniform
oscillations.

Finally, the four initial conditions chosen were:
\begin{align}\nonumber
  u_1(\theta) &= 1 &
  u_2(\theta)&=\frac{1}{\left|\theta-\frac{\pi}{3}\right|^{\nicefrac{1}{2}}}
  \\
  u_3(\theta)&=\mathbf{1}_{\left(\frac{7\pi}{8},\pi\right)}
    +\mathbf{1}_{\left(-\frac{3\pi}{4},-\frac{2\pi}{3}\right)} &
  u_4(\theta)&=
  \begin{cases}
    -\frac{\theta}{\pi} & \text{if}\;\theta<0 \\
    1-\frac{\theta}{\pi} & \text{if}\;\theta\geq0
  \end{cases}
\end{align}
They were chosen to include a mix of symmetric, non-symmetric, smooth and
non-smooth functions. Also, $u_2$ was chosen so that
$u_2\in(\Lp^1\backslash\Lp^2)$.

\begin{figure}
  \begin{minipage}{0.48\linewidth}
    A)

    \begin{center}
      \includegraphics[width=0.9\linewidth]{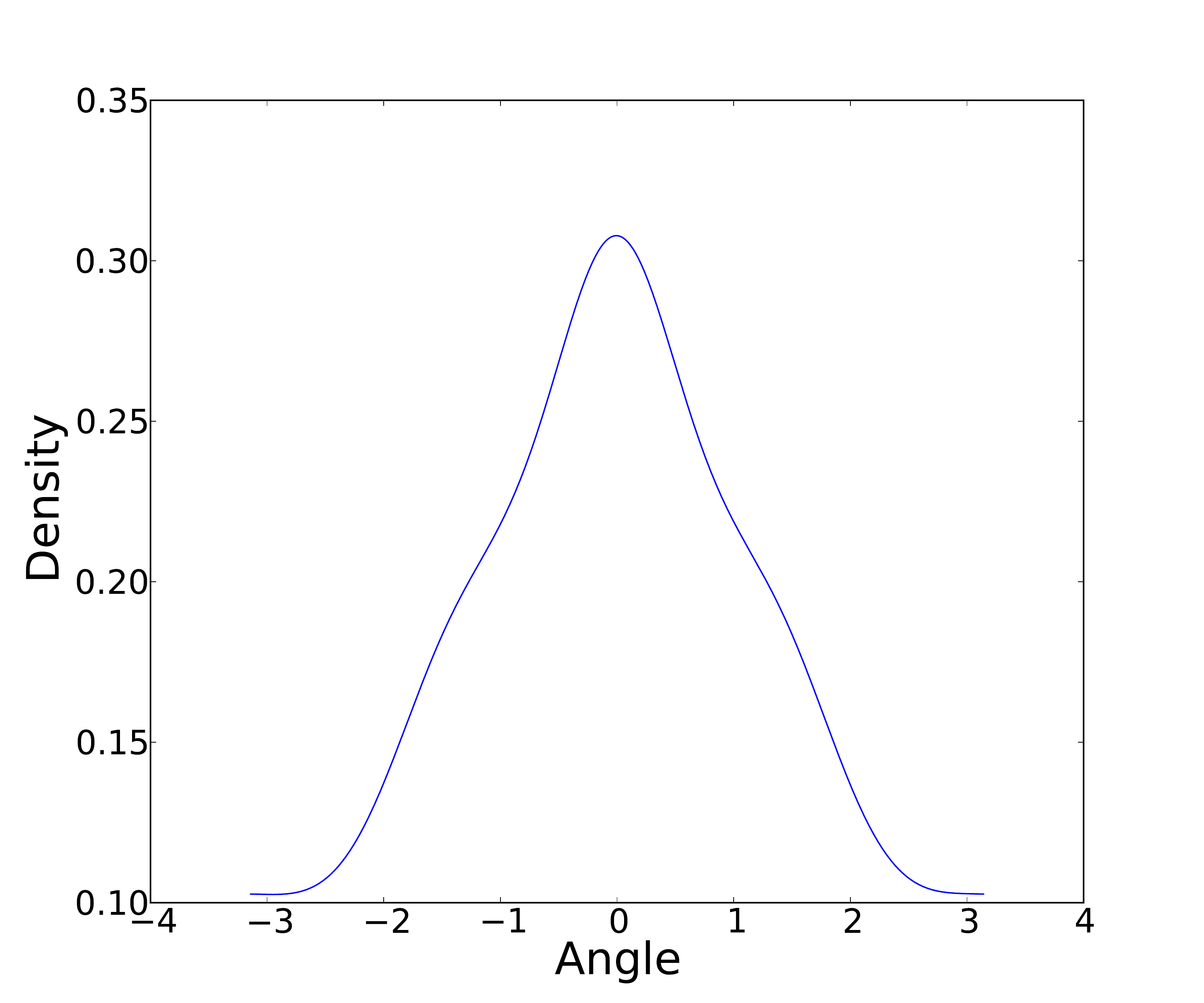}
    \end{center}
  \end{minipage}
  \begin{minipage}{0.48\linewidth}
    B)

    \includegraphics[width=0.9\linewidth]{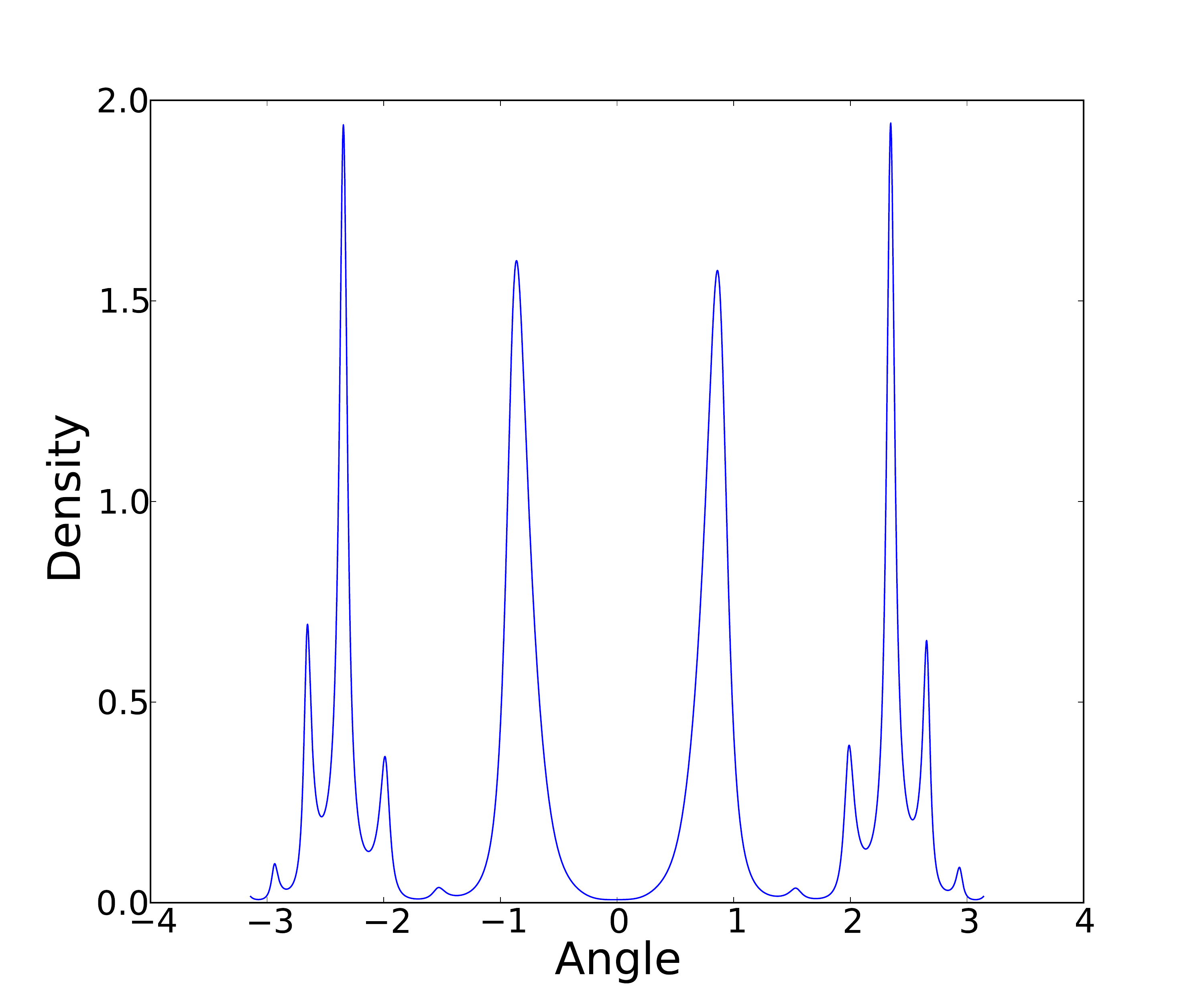}
  \end{minipage}

  \begin{minipage}{0.48\linewidth}
    C)

    \begin{center}
      \includegraphics[width=0.9\linewidth]{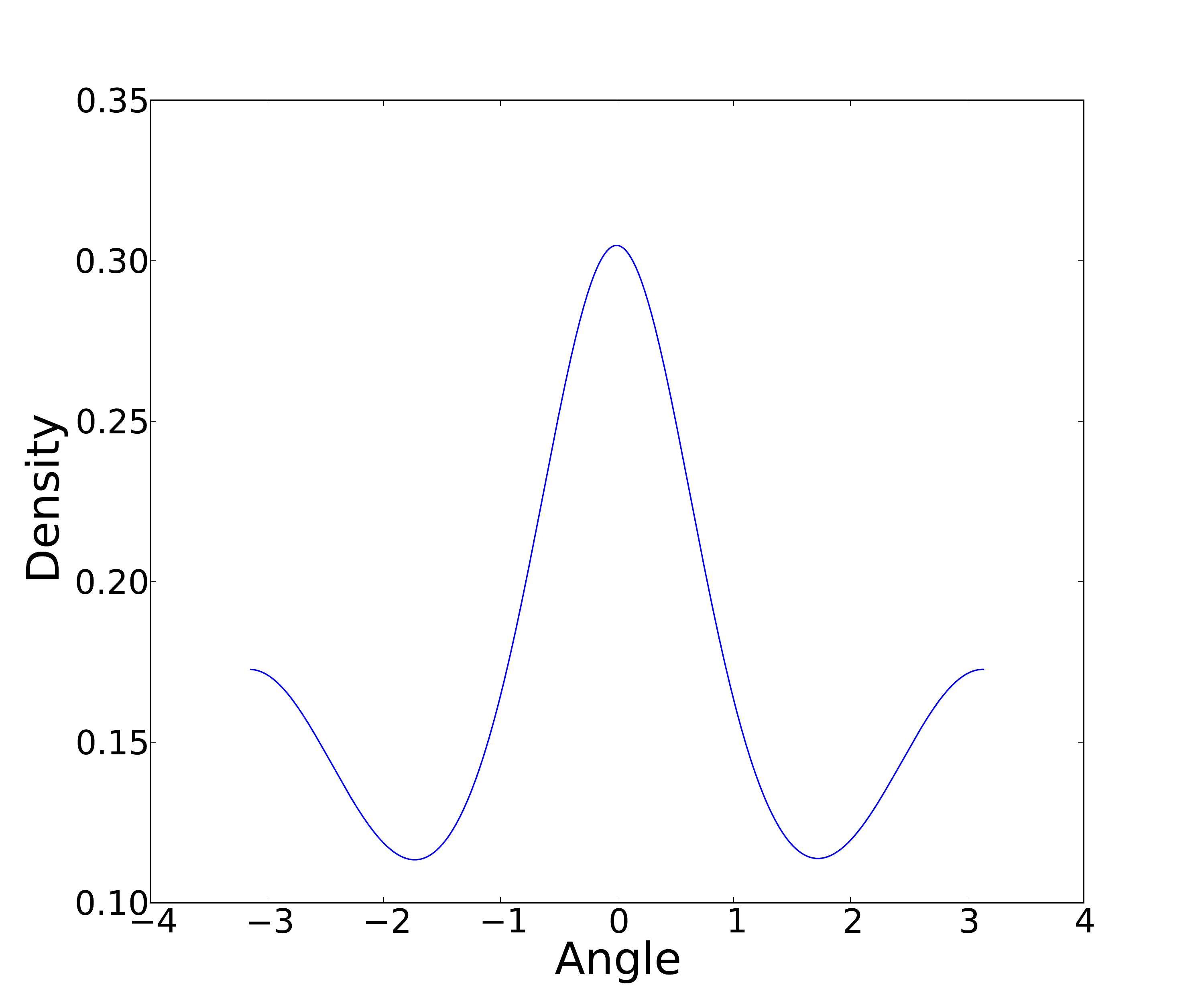}
    \end{center}
  \end{minipage}
  \begin{minipage}{0.48\linewidth}
    D)

    \begin{center}
      \includegraphics[width=0.9\linewidth]{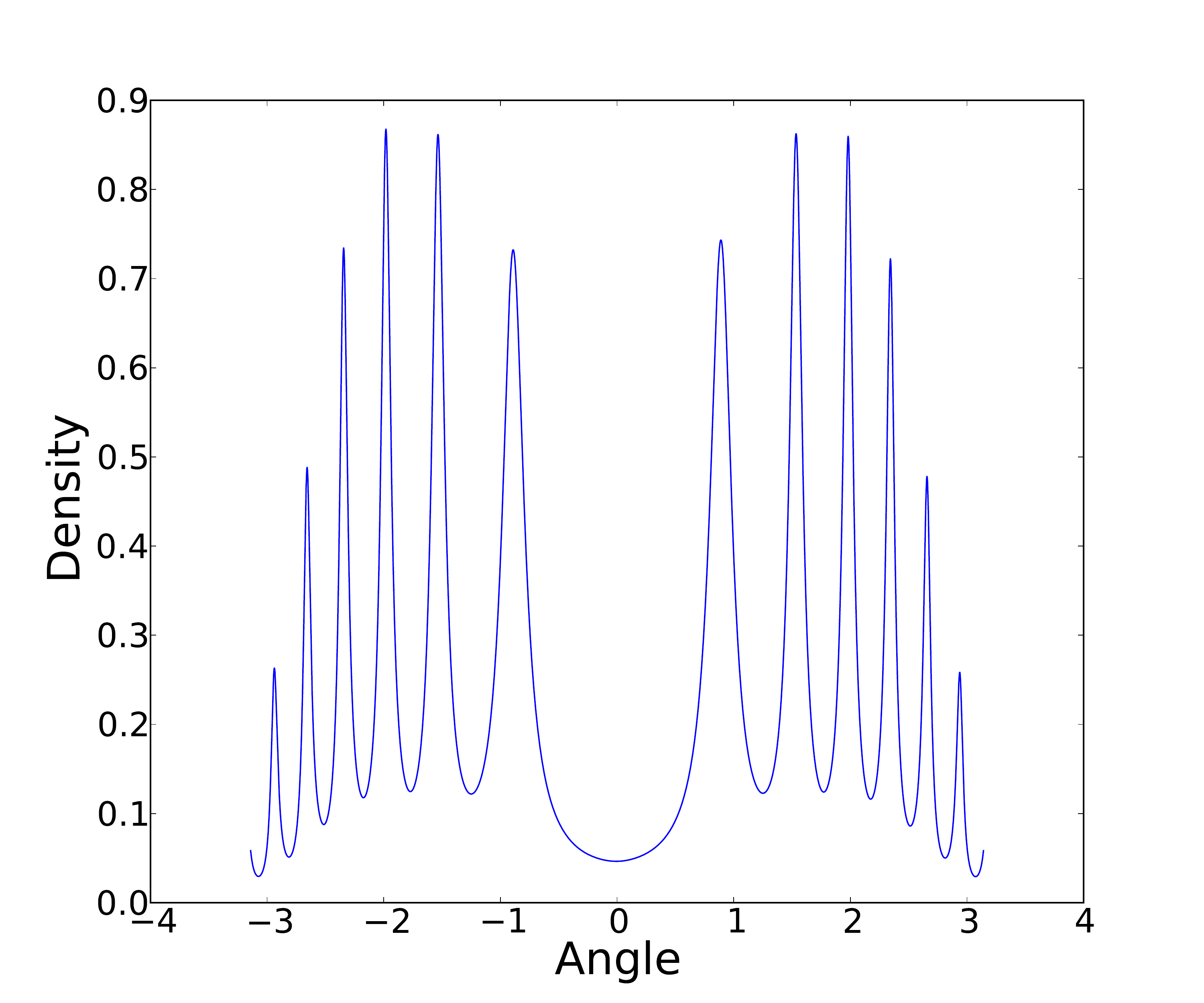}
    \end{center}
  \end{minipage}
  \caption[Equilibrium orientation distributions]{The plots here show the
    equilibrium distributions calculated using the method in section
    \ref{calc-equib} for the four systems. A) $\B_1$ and $\kappa_1$, B) $\B_1$
    and $\kappa_2$, C) $\B_2$ and $\kappa_1$, and D) $\B_2$ and
    $\kappa_2$.}\label{equib-dens-fig}
\end{figure}

The first calculations run were to estimate the equilibrium distribution using
the method in section \ref{calc-equib}. The method was iterated until the
$\|\cdot\|_1$ difference between successive iterations was less than double
precision.  The equilibrium distribution appeared to have a qualitatively
stronger dependence on $\kappa$ than on $\B$ as can be seen in Figure
\ref{equib-dens-fig}.

\begin{figure}[h]
  \begin{minipage}{0.48\linewidth}
    A)

    \includegraphics[width=\linewidth]{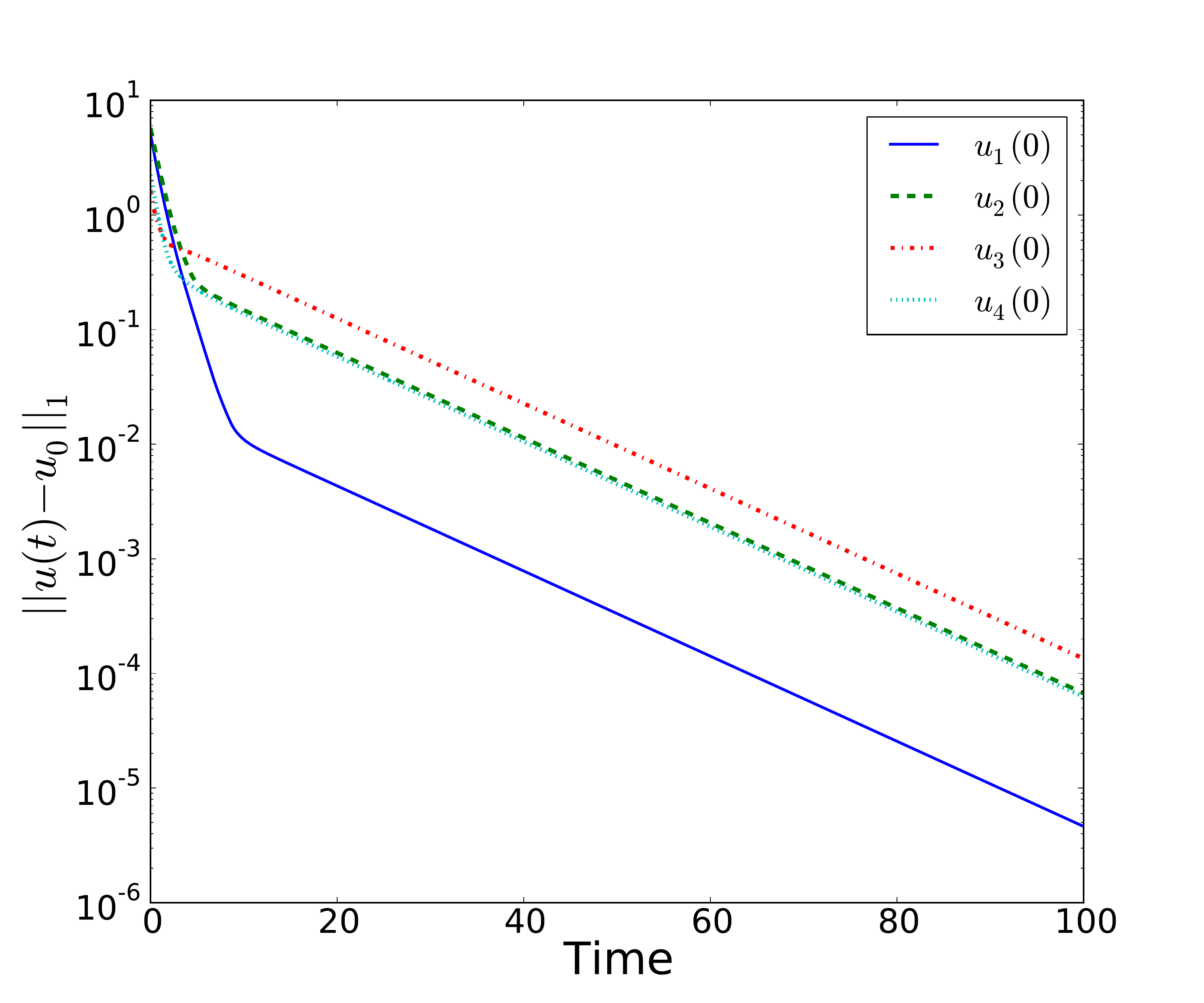}
  \end{minipage}
  \begin{minipage}{0.48\linewidth}
    B)

    \includegraphics[width=\linewidth]{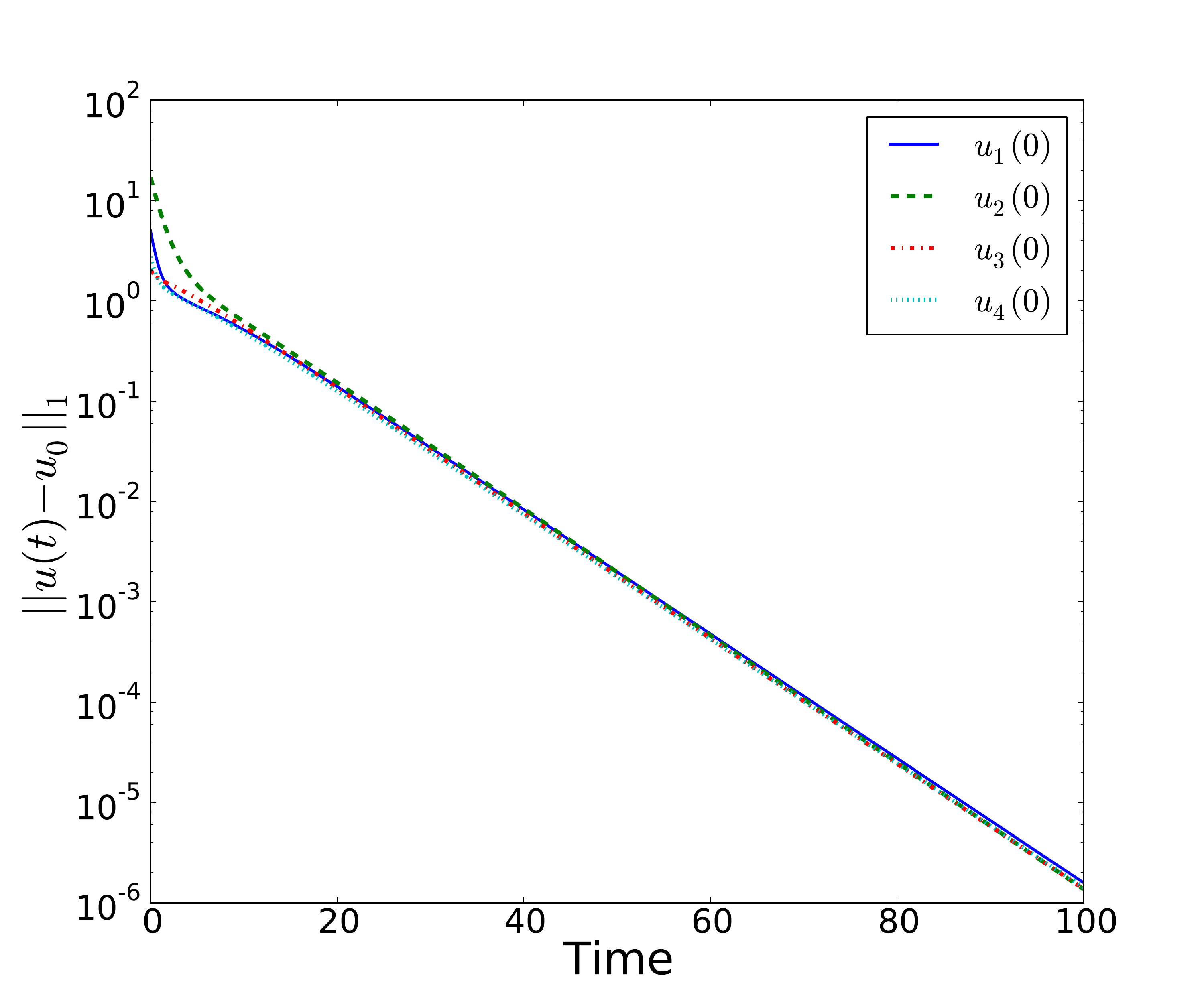}
  \end{minipage}

  \begin{minipage}{0.48\linewidth}
    C)

    \includegraphics[width=\linewidth]{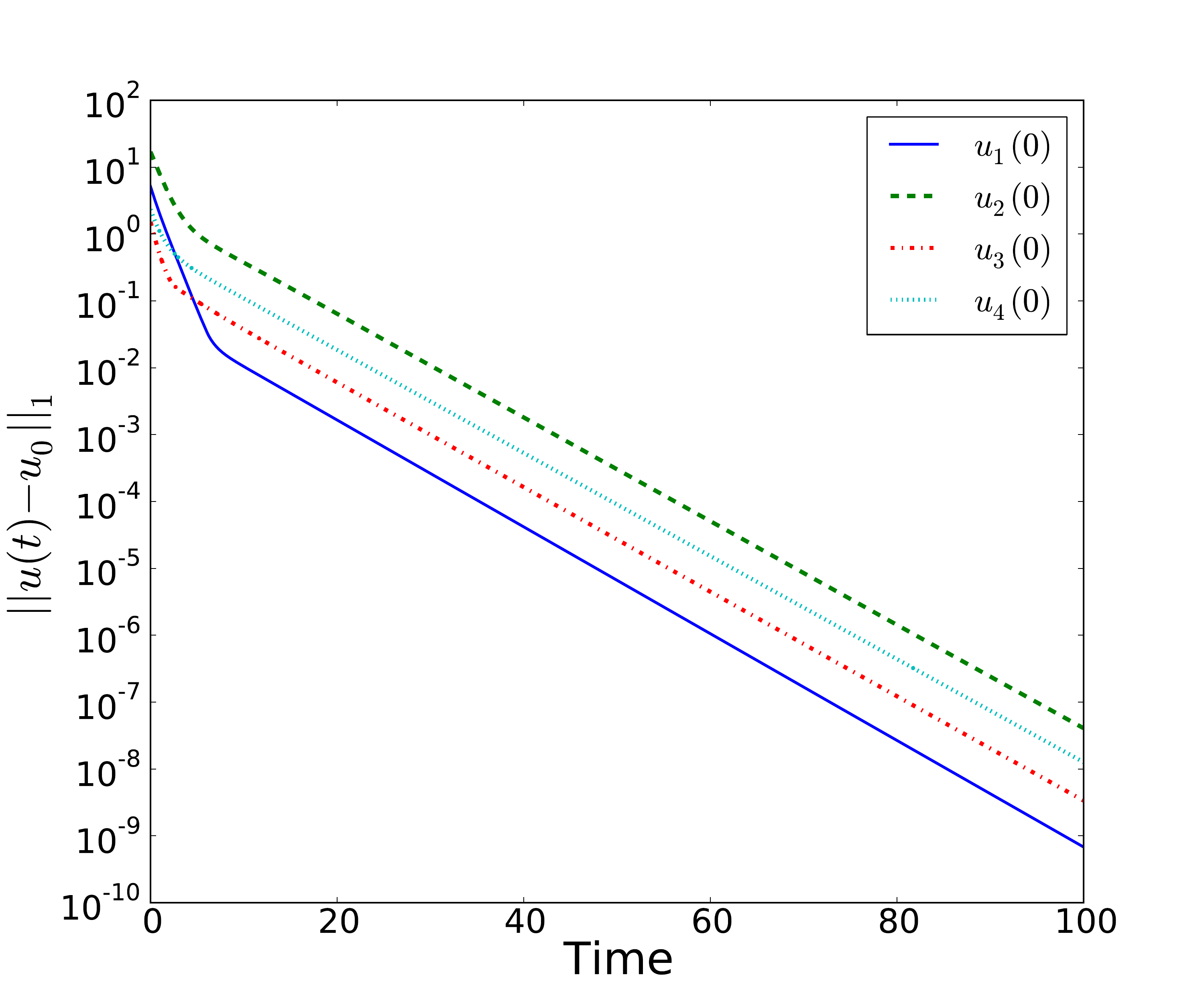}
  \end{minipage}
  \begin{minipage}{0.48\linewidth}
    D)

    \includegraphics[width=\linewidth]{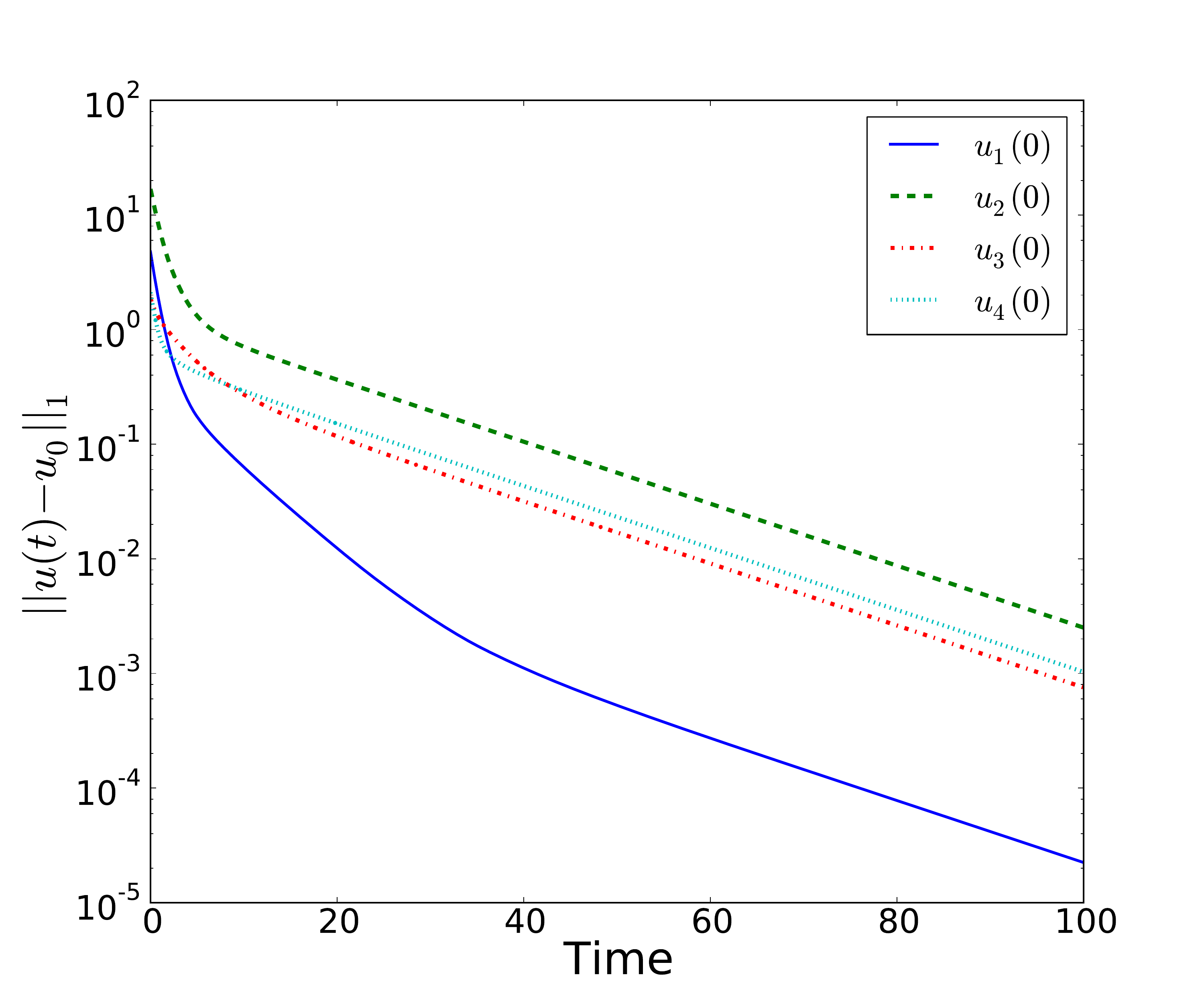}
  \end{minipage}
  \caption[The observed $\Lp^1$ distances from equilibrium as a function of
    time]{The observed $\Lp^1$ distances from equilibrium as a function of
    time are plotted here on a log scale. The four lines on each plot each
    correspond to the solution starting with a different initial condition. A)
    $\B_1$ and $\kappa_1$. B) $\B_1$ and $\kappa_2$. C) $\B_2$ and
    $\kappa_1$. D) $\B_2$ and $\kappa_2$.}\label{simulations}
\end{figure}

All of the simulations run converged (asymptotically) exponentially to the
equilibrium. The equilibrium was calculated by iterating $\op{A}^\prime$ as
outlined in the following section. Figure \ref{simulations} shows the $\Lp^1$
distance between the simulation result and the calculated equilibrium on a log
scale. The log scale was used to make the graphs legible and to show the
exponential convergence. Moreover, all of the initial conditions appear to
asymptotically converge at the same rate.  That gives evidence that solutions
are globally exponentially stable. The simulations provide evidence for
exponential stability with some constant determined solely by $\B$ and
$\kappa$.

\subsection{Perturbation Expansion}

The last approximate technique we will use is to study the perturbation
expanion of the zeroth-order branching equation. For this section, we will
again consider functions $u\in\Lp^1$. However, we will again appeal to Hilbert
space techniques when necessary.

Assume that the capping function can be written out as:
\begin{equation}
  \kappa(\theta)=c+\eps\phi(\theta)
\end{equation}
where $\phi$ is smooth, has integral zero, and reasonably small so that
$\eps\phi$ is close to zero. The equation of motion is thus:
\begin{equation}\label{pert-dyna}
  \dot{u}(\theta,t)=\frac{\Big(\B\star u\Big)(\theta,t)}
      {\int u(\omega,t)\dint\omega}-(c+\eps\phi(\theta))u(\omega,t)
\end{equation}
For the Cauchy problem with $u(\theta,0)=u^\star(\theta)$, we will look
for solutions of the form:
\begin{equation}
  u(\theta,t)=\sum_{j=0}^\infty\eps^ju_j(\theta,t)
\end{equation}
with the initial conditions $u_0(\theta,0)=u^\star$ and $u_j(\theta,0)=0$ for
all $j\geq1$. Showing that both $u_0$ and $u_1$ converge to the equilibribium
defined in equation \eqref{zero-eigen} provides some evidence for the stability
of the full equation.

First, we will calculate the first couple of terms of the equilibrium using
equation \eqref{zero-eigen}:
\begin{equation}\label{pert-eigen}
  \frac{\Big(\B\star w\Big)(\theta)}{1+\eps\phi(\theta)}=\mu w(\theta)
\end{equation}
where we have the two expansions:
\begin{equation}
  w(\theta)=\sum_{j=0}^\infty\eps^jw_j(\theta)
  \qquad
  \mu=\sum_{j=0}^\infty\eps^j\mu_j
\end{equation}
Expanding out the terms in \eqref{pert-eigen} gives:
\begin{align}\nonumber
  \left(\sum_{j=0}^\infty\frac{(-\eps\phi)^j}{c^{j+1}}\right)
  \left(\sum_{j=0}^\infty\eps^j\B\star w_j\right)
  &=\left(\sum_{j=0}^\infty\eps^j\mu_j\right)
  \left(\sum_{j=0}^\infty\eps^jw_j\right) \\\label{expand-pert-eigen}
  \sum_{j=0}^\infty\eps^j\sum_{k=0}^j\frac{(-\phi)^k}{c^{k+1}}\B\star w_{j-k}
  &=\sum_{j=0}^\infty\eps^j\sum_{k=0}^j\mu_jw_{j-k}
\end{align}
The first term ($\eps^0$) is simply the eigenproblem for the unperturbed
problem:
\begin{equation}
  \B\star w_0=c\mu_0w_0
\end{equation}
The only positive solution of the above is $\mu_0=\frac{1}{c}$ and $w_0$
constant. Since we know the equilibrium has the same integral as the
eigenvalue, we know that $w_0=\frac{1}{c2\pi}$. The $\eps^1$ term gives:
\begin{equation}
  \frac{\B\star w_1}{c}-\frac{\phi\B\star w_0}{c^2}=\mu_0w_1+\mu_1w_0
\end{equation}
Filling in the known quantities and rearranging gives:
\begin{equation}\label{fred}
  \B\star w_1-w_1=\frac{\phi}{c}+\frac{\mu_1}{2\pi}
\end{equation}
The left hand side is a self-adjoint Fredholm operator with nullspace spanned
by the constant function, $w_0$. By the Fredholm alternative, we know that
\eqref{fred} is solvable if and only if the right hand side is orthogonal to
$w_0$, i.e. $\phi(\theta)+\frac{\mu_1}{2\pi}$ has integral zero. To have
integral zero, we know that
$\mu_1=-\frac{2\pi}{c}\int\phi(\omega)\dint\omega=0$.

We can write out $w_1$ in terms of the Neumann series. By hypothesis, $\phi$ is
bounded, and therefore $\phi\in\Lp^2$.  We know that $\B\star$ has operator
norm strictly less than one on the space of functions orthogonal to $w_0$:
$\{w_0\}^\perp\subsetneq\Lp^2$, which implies the Neumann series converges in
$\Lp^2$ norm:
\begin{equation}
  w_1=\sum_{j=0}^\infty\op{B}^j\phi
\end{equation}
where $\op{B}u=\B\star u$ and $\op{B}^{j+1}u=\B\star\op{B}^ju$. The above
expansion implies that $\int w_1=0$. We know that $\int\op{B}u=\int u$ which
gives:
\begin{equation}
  \int\sum_{j=0}^n\Big(\op{B}^j\phi\Big)(\theta)\dint\theta
  =0
\end{equation}
Since norm convergence implies weak convergence, we can finish the proof by
observing:
\begin{equation}
  \lim_{n\to\infty}\int w_1-\int\sum_{j=0}^n\op{B}^j\phi
  =\lim_{n\to\infty}\left\langle w_1(\theta)
  -\sum_{j=0}^n\op{B}^j\phi,1\right\rangle=0
\end{equation}

We can now consider the dynamics. We can write out \eqref{pert-dyna} in terms
of our power series:
\begin{multline}
  \sum_{j=0}^\infty\eps^j\dot{u}_j(\theta,t)
  =\frac{1}{\sum_{j=0}^\infty\eps^j\int u_j(\omega,t)\dint\omega}
  \sum_{j=0}^\infty\eps^j\Big(\B\star u_j\Big)(\theta,t) \\
  -(c+\eps\phi(\theta))\sum_{j=0}^\infty\eps^ju(\theta,t)
\end{multline}
A formal treatment of the integrand could be considered, but without confidence
of convergence, that seems unnecessary. We are only calculating $u_0$ and
$u_1$, so we will ignore terms of $o(\eps^2)$ or higher in the integrand:
\begin{multline}\label{pert-solve}
  \sum_{j=0}^\infty\eps^j\dot{u}_j(\theta,t)
  =\left(\sum_{j=0}^\infty\frac{\left(-\eps\int u_1(\omega,t)
    \dint\omega\right)^j}
  {\left(\int u_0(\omega,t)\dint\omega\right)^{j+1}}\right)
  \left(\sum_{j=0}^\infty\eps^j\Big(\B\star u_j\Big)(\theta,t)\right) \\
  -\Big(c+\eps\phi(\theta)\Big)\sum_{j=0}^\infty\eps^ju_j(\theta,t)
\end{multline}
The above equation allows us to solve for the first two terms of the
perturbation expansion.

The equation for the first term $u_0(\theta,t)$ from the $\eps^0$ expansion is:
\begin{equation}\label{first-pert}
  \dot{u}_0(\theta,t)=
  \frac{\Big(\B\star u_0\Big)(\theta,t)}{\int u_0(\omega,t)\dint\omega}
  -u_0(\theta,t)
\end{equation}
We can explicitly solve for the time-dependent total density by observing:
\begin{equation}
  \frac{\partial}{\partial t}\int u(\omega,t)\dint\omega=
  1-\int u(\omega,t)\dint\omega
\end{equation}
Solving the above gives:
\begin{equation}\label{asympt-dens}
  \int u(\omega,t)\dint\omega=1+A\exp[-t]
\end{equation}
where $A=\int u^\star(\omega)\dint\omega$. Substituting in to \eqref{first-pert}
gives:
\begin{equation}
  \dot{u}_0(\theta,t)=\frac{\Big(\B\star u_0\Big)(\theta,t)}{1+A\exp[-t]}
  -u_0(\theta,t)
\end{equation}
The right hand side of the above clearly depends upon the denominator
$1+A\exp[-t]$ continuously in almost any operator topology. Since we are
primarily concerned with the asymptotic dynamics, it is sufficient to show that
the asymptotic equation converges:
\begin{equation}
  \dot{u}_0=\B\star u_0-u_0
\end{equation}
From the results in Theorem \ref{principal}, we know that the above equation
converges to a multiple of the principal eigenfunction $w_0$. Equation
\eqref{asympt-dens} implies that the total density is equal to $\frac{1}{c}$
and $u_0(\theta,t)\to w_0(\theta)$ in norm as was required.

The second term of the expansion $(\eps^1)$ is slightly more complicated:
\begin{equation}\label{second-pert}
  \dot{u}_1(\theta,t)=
  \frac{\Big(\B\star u_1\Big)(\theta,t)}{\int u_0(\omega,t)\dint\omega}
  -\frac{\int u_1(\omega,t)\dint\omega}
  {\left(\int u_0(\omega,t)\dint\omega\right)^2}\Big(\B\star u_0\Big)(\theta,t)
  -u_1(\theta,t)-\phi(\theta)u_0(\theta,t)
\end{equation}
Substituting in known quantities gives:
\begin{equation}
  \dot{u}_1=\frac{\B\star u_1}{1+A\exp[-t]}-\frac{\int u_1}{(1+A\exp[-t])^2}
  \B\star u_0-u_1-\phi u_0
\end{equation}
We can now solve for the integral of $u_1$:
\begin{align}\nonumber
  \frac{\partial}{\partial t}\int u_1&=
  \frac{\int\B\star u_1}{1+A\exp[-t]}-\frac{\int u_1}{1+A\exp[-t]}
  -\int u_1-\int \phi u_0 \\
  &=-\int u_1-\int \phi u_0
\end{align}
That gives the explicit solution:
\begin{equation}
  \int u_1(\omega,t)\dint\omega=-\exp[-t]
  \int\limits_0^t\exp[s]\int\phi(\omega)u_0(\omega,t)\dint\omega\dint s
\end{equation}
Weak convergence of $u_0\to w_0$ is sufficient to show that $\int u_1\to-\int
\phi$. Substituting the asymptotic forms into \eqref{second-pert} similar to
the first term gives:
\begin{equation}\label{second-asympt}
  \dot{u}_1=\B\star u_1-u_1+\int\phi-\frac{\phi}{2\pi}
\end{equation}
asymptotically. Finally, to show convergence, we can write $u_1$ as:
\begin{equation}
  u_1(\theta,t)=w_1(\theta)+\epsilon(\theta,t)
\end{equation}
where $\epsilon(\theta,0)=-w_1(\theta)$. Putting that form into
\eqref{second-asympt} gives:
\begin{equation}
  \dot{\epsilon}=\B\star\epsilon-\epsilon
\end{equation}
Inspection of \eqref{fred} shows that $w_1(\theta)$ is bounded and is therefore
in $\Lp^2$. Since $\B\star u-u$ has the Fourier eigenpairs
$\{(\gamma_n,\frac{1}{\sqrt{2\pi}}\exp[-in\theta])\}_{n\in\Nat}$ as an
orthonormal basis, we can write the $\Lp^2$ norm of $\eps$ as:
\begin{equation}
  \|\epsilon(\theta,t)\|_2^2=\sum_{n\in\Nat}\exp[\gamma_nt]|\hat{w}_1(n)|^2
\end{equation}
where $\hat{w}_1(n)$ is the $n$'th Fourier component of $w_1(\theta)$. We know
that $\gamma_n<0$ for all $n>0$\footnote{the $\gamma_n$'s are real since $\B$ is
  symmetric.}  and $\sup_{n\geq1}\gamma_n<0$. Combining those facts with the fact
that $\hat{w}_1(0)=0$ because $\int w_1=0$ gives:
\begin{equation}
  \|\epsilon(\theta,t)\|_2^2=\sum_{n\geq1}\exp[\gamma_nt]|\hat{w}_1(n)|^2
  \leq\exp\left[\left(\sup_{n\geq1}\gamma_n\right)t\right]
  \sum_{n\geq1}|\hat{w}_1(n)|^2
\end{equation}
That implies that $u_1(\theta,t)\to w_1(\theta)$ asymptotically exponentially.

We have now shown that the first two components of the perturbation expansion
of \eqref{pert-dyna} converge to an equilibrium. Calculating further terms
would not provide any additional insight as the estimate from equation
\eqref{expand-pert-eigen} gives a worse estimate after the first order. Recall
from the proof of Proposition \ref{exist-prop} that the integral of the
equilibrium $w$ has to be equal to the eigenvalue $\mu$. However, it is easy to
observe that the argument showing that $\int w_1=$ holds for all $w_j$, which
implies that $\int \sum_{j=0}^k\eps^jw_j = \frac{1}{c}$ for all $k$. However,
$\sum_{j=0}^k\eps^j\mu_j\neq\mu_0=\frac{1}{c}$ for all $k\geq2$. We can show
that by caclulating the second term in the perturbation expansion.

Using equation \eqref{expand-pert-eigen}, we can gather the $\eps^2$ terms:
\begin{equation}
  \frac{1}{c}\B\star w_2-\frac{\phi}{c^2}\B\star w_1
  +\frac{\phi^2}{c^3}\B\star w_0
  =\mu_0w_2+\mu_1w_1+\mu_2w_0
\end{equation}
Observe that $\B\star w_1=w_1+\frac{\phi}{c}$, and filling in other known
quantities gives the relation:
\begin{align}\nonumber
  \frac{1}{c}\left(\B\star w_2-w_2\right)&=\frac{\phi}{c^2}\B\star w_1
  -\frac{\phi^2}{c^3}+\mu_2w_0 \\\nonumber
  \B\star w_2-w_2&=\frac{\phi^2}{c^2}-\frac{\phi^2}{c^2}+\frac{\phi w_1}{c}
  +\frac{\mu_2}{2\pi} \\
  \B\star w_2-w_2&=\frac{\phi w_1}{c}+\frac{\mu_2}{2\pi}
\end{align}
It is obvious that $\mu_2=0\iff\int\phi w_1=0$. However, we can observe that:
\begin{align}\nonumber
  \int \frac{\phi w_1}{c}=\int \phi\left(\op{B}-\op{I}\right)^{-1}\frac{\phi}{c}
  &=\left\langle\phi,\left(\op{B}-\op{I}\right)^{-1}\frac{\phi}{c}\right\rangle
  \\\nonumber
  &=\Big\langle(\op{B}-\op{I})w_1,
  (\op{B}-\op{I})(\op{B}-\op{I})^{-1}w_1\Big\rangle \\
  &=\Big\langle(\op{B}-\op{I})w_1,w_1\Big\rangle
\end{align}
by observing that $(\op{B}-\op{I})w_1=\frac{\phi}{c}$. We can use the
Cauchy-Schwarz inequality to show:
\begin{equation}
  \langle\op{B}w_1,w_1\rangle-\langle w_1,w_1\rangle
  \leq\|\op{B}w_1\|_2\|w_1\|_2-\|w_1\|_2^2<0
\end{equation}
The last inequality comes from the fact that
$\|\op{B}f\|_2\leq\|\op{B}\|\|f\|_2=\|f\|_2$ and that equality holds if and
only if $f\equiv C$.

An explicit example can give insight into the non-convergence discussed in the
previous paragraph. Assume $\B=\frac{1}{2\pi}$ and $\kappa = 1+\eps\phi$ where
$|\eps\phi|<1$. It is easy to see that the principal solution to
\eqref{pert-eigen} is $w=\frac{1}{1+\eps\phi}$ and
$\mu=\int\frac{1}{1+\eps\phi}$. Moreover, we have the exact expansion in
$\eps$:
\begin{equation}
  w=\sum_{j=0}^\infty\eps^jw_j=\sum_{j=0}^\infty\eps^2(-\phi)^j
\end{equation}
By our hypothesis $|\eps\phi|<1$, we know the above converges. It is easy to
see that:
\begin{equation}
  (\op{B}-\op{I})^{-1}\phi=-\sum_{j=0}^\infty\op{B}^j\phi=-\phi
\end{equation}
which is the exact solution for the first term of the asymptotic expansion.
However, the second term gives an incorrect solution:
\begin{align}\nonumber
  (\op{B}-\op{I})^{-1}\left(\phi^2-\int\phi^2\right)&=
  \sum_{j=0}^\infty\op{B}\left(\int\phi^2-\phi^2\right) \\
  &=\int\phi^2(\omega)\dint\omega-\phi^2(\theta)\neq\phi^2(\theta)
\end{align}
Further terms would show the same difficulty.

\section{Conclusions}

The results presented here provide a reference point for future work on the
orientation patterns of branching actin networks. Given the current models, it
would seem that factors external to the actin network determine the distinct
orientation patterns observed in experiment. In experiments where the cell goes
through protrusion/retraction cycles, different orientation patterns are
observed at different points in the cycle \cite{Giannone2004}. One possible
factor leading to the varying patterns may be how the network deals with the
load from the cell membrane \cite{Smith2013}.

In the numerical results, only $\B_1$ and $\kappa_1$ were physically based, but
Figure \ref{equib-dens-fig} show that there is a complex interplay between the
capping and branching functions to result in the equilibrium distribution.  The
stability seen in Figure \ref{simulations} may indicate why the orientation
patterns seen in experiment have been so stable. The numerical results also
reinfoce the analytical result that there is only one stable orientation
pattern, in contrast with the multiple equilibria hypothesis in Weichsel and
Schwarz \cite{Weichsel2010}.

\bibliographystyle{siam}
\bibliography{orient}

\appendix
\section{Methods}

All simulations were run using the Numpy \cite{Ascher1999, Oliphant2006,
  Oliphant2007, Scipy} extension to Python. The equations of motion were
integrated using a simple Euler method. The simulations were run for 100 time
units with a time step of 0.01 time units. The circle was discretized using
$2^{11}$ equally spaced points from $-\pi$ to $\pi$ (the power of 2 was used to
speed up the fast Fourier transform). All integrals were taken using the
trapezoidal method included in Numpy. The convolution was performed by taking
the real fast Fourier transform of the branching kernel $\B$ and the density
$u(\theta,t)$, multiplying, and taking the inverse real fast Fourier transform.
The built-in Numpy fast convolution method was not used because that method
pads the two convolved functions with extra zeros to prevent circular
convolution, but the equations used here explicitly call for the circular
convolution. The convolution was normalized by dividing $\B$ by the integral of
the convolution of $\B$ with the constant function
$\B\star\frac{1}{2\pi}$. Finally, the total branching rate was normalized by
integrating the density at the previous time step, i.e.:
\begin{equation}
  u(\theta,t+\dint t) = \dint t\left(
  \frac{\Big(\B\star u\Big)(\theta,t)}{\int u(\theta,t)\dint\theta}
  -\kappa(\theta)u(\theta,t)\right)
\end{equation}
Using the integral from the previous time step and not a more sophisticated
prediction-correction methods is justified by the following inequality:
\begin{align}\nonumber
  |u(\theta,t+\dint t)-u(\theta,t)|&\leq \dint t\left(
  \int\Big(\B\star u\Big)(\theta,t)\dint\theta
  +\int\kappa(\theta)u(\theta,t)\dint\theta\right) \\
  &\leq\Big(1+\sup_{\theta\in S^1}\kappa(\theta)\Big)\dint t\|u(\theta,t)\|_1
  \leq3\dint t\|u(\theta,t)\|_1
\end{align}
Since simulations remain bounded, the bound above can be made uniform.

The equilibrium was calculated a priori by iterating the equilibrium operator
$\op{A}^\prime$. Explicitly, a sequence of functions was generated by:
\begin{equation}
  v_{n+1}(\theta)=\frac{1}{\int v_n(\theta)\dint\theta}\op{A}^\prime v_n(\theta)
\end{equation}
where the discretization and convolution were performed exactly as above and
$v_0=\frac{1}{2\pi}$. The theoretical justification for using this method is
outlined in Section \ref{calc-equib}. For three combinations of $\B$ and
$\kappa$, $\|v_{10^4}-v_{2\times10^4}\|_1$ was less than numerical
precision. For the combination of $\B_1$ and $\kappa_2$, $v_{10^4}$ was not
sufficiently converged, so $v_{10^6}$ was used. That decision was based on the
condition that $\|v_{10^6}-v_{2\times10^6}\|_1$ was less than numerical
precision. That level of precision was used to ensure that the convergence
could be seen even when $\|u(\theta,t)-v_n\|_1\leq10^{-6}-10^{-10}$. Even
$10^6$ iterations of $\op{A}^\prime$ only took several minutes on a standard
Linux desktop system concurrently running other programs, a number that could
be reduced with further optimization.

At each time-step, the $\Lp^1$ distance between the state of the system
$u(\theta,t)$ and the equilibrium was calculated. That quantity is plotted as a
function of time in Figure \ref{simulations}.

\section{Calculating Equilibrium Distributions}\label{calc-equib}

The results in this article justify the use of a naive eigenvalue calculation
algorithm. Calculating eigenvalues of integral equations is a non-trivial
problem. Investigation into open questions regarding the generality of
orientation patterns across branching and capping patterns, such as in
\cite{Quint2011}, may require calculating the equilibrium solution to equations
like the ones analyzed here. Moreover, in the previous section, equilibrium
distributions were calculated a priori to show that simulations converged. The
method below has proven to be very efficient for the work in this article.

We will consider calculating the leading eigenvalue of the equilibrium operator
for zeroth-order branching, $\op{A}^\prime$. As $\op{A}^\prime$ is self-adjoint
and compact, we can represent its range as the sum of eigenfunctions. We can
explicitly calculate the $n$-th iterate of $\op{A}^\prime$ in terms of its
(orthonormal) eigenfunctions:
\begin{equation}
  \Big(\op{A}^{\prime}\Big)^nv
  =\Big(\op{A}^{\prime}\Big)^n\Big(\sum_jc_ju_j\Big)
  =\sum_j\mu_j^nc_ju_j
\end{equation}
where $c_j=\langle v,u_j\rangle_\kappa$. We know that $\mu_0$ is equal to the
spectral radius from the Krein-Rutman theorem as in the proof of
Theorem \ref{principal}. The proof also implies the eigenvalue is
simple. Finally, since $\op{A}^\prime$ is a compact operator, we know that
there must be a spectral gap, i.e. $\mu_0-|\mu_j|>c>0$ for some $c$ and all
$j\neq0$.

All that remains necessary to show that the above iteration converges to the
positive equilibrium is to show that $c_j=\langle v,u_0\rangle_\kappa\neq0$. If
$v$ equals the constant function, that condition is fulfilled. However, a
stronger result is possible. By a result in \cite{DeMasi1998}, we know that
$u_0(\theta)>0$. The continuity of $u_0$ gives that $\inf u_0>0$. Thus, we have
the inequality:
\begin{equation}
  \langle v,u_0\rangle_\kappa\geq\inf_{\theta\in S^1}(u_0(\theta)\kappa(\theta))
  \int\limits_{S^1}v(\theta)\dint\theta>0
\end{equation}
which implies that the iterative procedure will converge.

\end{document}